\documentclass{llncs}
\pdfoutput=1 
\newif\ifspace\spacefalse

\newif\ifsinglecolumn\singlecolumntrue

\usepackage{xspace}
\usepackage{subfigure}
\usepackage{graphicx}

 \usepackage[bookmarks=false]{hyperref}
\usepackage{algorithm}
\usepackage{algpseudocode}
\usepackage{wrapfig}

\newtheorem{thm}{Theorem} 

\newtheorem{lem}{Lemma}

\usepackage{array}

\newcommand{\systype}[1]{\ensuremath{\langle #1\rangle}\xspace}

\newcommand{\br}[1]{\ensuremath{\langle #1 \rangle}\xspace}

\newcommand{\epone}{\ensuremath{\epsilon_{p_{1}}}\xspace}
\newcommand{\eptwo}{\ensuremath{\epsilon_{p_{2}}}\xspace}

\newcommand{\hb}{\ hb\xspace}
\newcommand{\acbound}{\eta \xspace}

\newcommand{\ineffective}{ineffective\xspace}
\newcommand{\uncertain}{uncertain\xspace}
\newcommand{\functional}{functional\xspace}

\graphicspath{{.}{./Figures-eps/}}

\title{Precision, Recall, and Sensitivity of Monitoring Partially Synchronous Distributed Systems} 

\author{ Sorrachai Yingchareonthawornchai\inst{1}, Duong Nguyen\inst{1}, Vidhya Tekken Valapil\inst{1}, Sandeep Kulkarni\inst{1}, and Murat Demirbas\inst{2}}

\institute{ Department of Computer Science and Engineering\\
Michigan State University\\
East Lansing MI 48824 \\
\email{\{yingchar, nguye476,  tekkenva, sandeep\}@cse.msu.edu}
\and
Department of Computer Science and Engineering\\
University at Buffalo, The State University of New York\\
Buffalo NY 14260-2500\\
\email{demirbas@cse.buffalo.edu}\\
}

\begin{document}
\maketitle

\begin{abstract}
Runtime verification focuses on analyzing the execution of a given program by a monitor to determine if it is likely to violate its specifications. There is often an impedance mismatch between the assumptions/model of the monitor and that of the underlying program. This constitutes problems especially for distributed systems, where the concept of current time and state are inherently uncertain. A monitor designed with asynchronous system model assumptions may cause false-positives for a program executing in a partially synchronous system: the monitor may flag a global predicate that does not actually occur in the underlying system. A monitor designed with a partially synchronous system model assumption may cause false negatives as well as false positives for a program executing in an environment where the bounds on partial synchrony differ (albeit temporarily) from the monitor model assumptions.

In this paper we analyze the effects of the impedance mismatch between the monitor and the underlying program for the detection of conjunctive predicates. We find that there is a small interval where the monitor assumptions are hypersensitive to the underlying program environment. We provide analytical derivations for this interval, and also provide simulation support for exploring the sensitivity of predicate detection to the impedance mismatch between the monitor and the program under a partially synchronous system.

\end{abstract}

\section{Introduction}
\label{sec:intro}

Runtime verification focuses on analyzing the execution of a given program by a monitor to determine if it violates its specifications. In analyzing a distributed program, the monitor needs to take into account multiple processes simultaneously to determine the possibility of violation of the specification. Unfortunately, perfect clock synchronization is unattainable for distributed systems~\cite{gradient,clockSync}, and distributed systems have an inherent uncertainty associated with the concept of current time and state~\cite{nonow}. As a result, there is often an impedance mismatch between the assumptions/model of the monitor and that of the underlying program. Even after a careful analysis of the underlying distributed system/program, the model assumptions that the monitor infers for the system/program will have errors due to uncertain communication latencies (especially over multihops over the Internet), temporal perturbations of clock synchronization (especially when different multihop clock references~\cite{ntp} are used), and faults.

In the absence of precise knowledge about events there is a potential that the
debugging/monitoring system (which we call as the monitor) would either (1) find
non-existent bugs or/and (2) miss existing bugs. While some error is
unavoidable, if we cannot characterize monitor and the underlying program/system
behavior precisely, there is no analysis to answer the effect of system
uncertainty on predicate detection/runtime verification. Our goal in this paper
is to analyze the errors caused by uncertainty of the underlying distributed
system and the impedance mismatch between the monitor and the underlying
distributed system.

To illustrate the role of the uncertainty and the impedance mismatch, consider
the example in Figure \ref{fig:effectoftimeout}. In this computation, we want to
verify that the system never reaches a state where the predicate $x > 0 \wedge y
> 0$ is true. In Figure \ref{fig:effectoftimeout} (a), it is clear that the
predicate is not true since there is a message after $x > 0$ has become false
and before $y > 0$ becomes true. In Figure \ref{fig:effectoftimeout} (b), if the
processes' clocks were perfectly synchronized the predicate is always false.
However, if it is assumed that the processes are asynchronous or can have large
clock drifts then in Figure \ref{fig:effectoftimeout} (b), the predicate is
true. In other words, if the algorithm for runtime monitoring assumes that the
system clock is perfectly synchronized but in reality it is not then in
Figure \ref{fig:effectoftimeout} (b), the result of the monitoring algorithm
will be false negative, i.e., the monitor will fail to detect that the system
(possibly) reached a state where $x > 0 \wedge y > 0$ was true. On the other
hand, if the monitoring algorithm assumes an asynchronous system but in reality,
it is synchronous (and the system may be using timeouts as implicit
communication) then in Figure \ref{fig:effectoftimeout} (b), the result of the
monitoring algorithm is false positive, i.e., the monitor incorrectly finds that
the system (possibly) reached a state where $x > 0 \wedge y > 0$ was true.

\begin{figure*}[tbhp]
\begin{minipage}{\linewidth}
\begin{center}
 \vspace{-10pt}
\subfigure[]{\includegraphics[width=0.4\columnwidth]{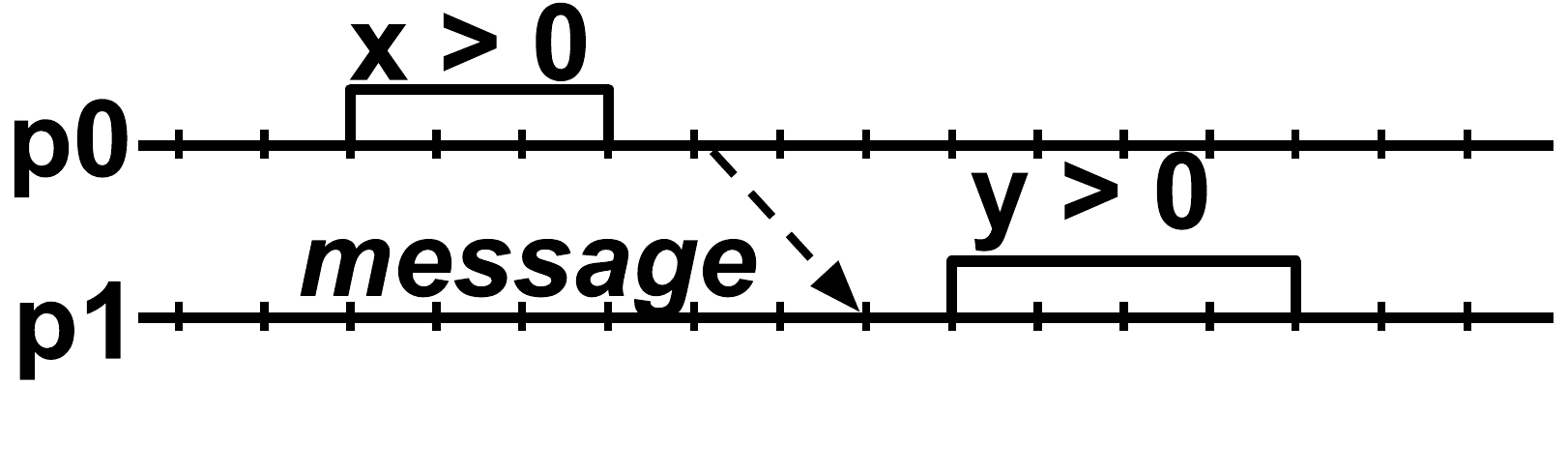}}
\subfigure[]{\includegraphics[width=0.4\columnwidth]{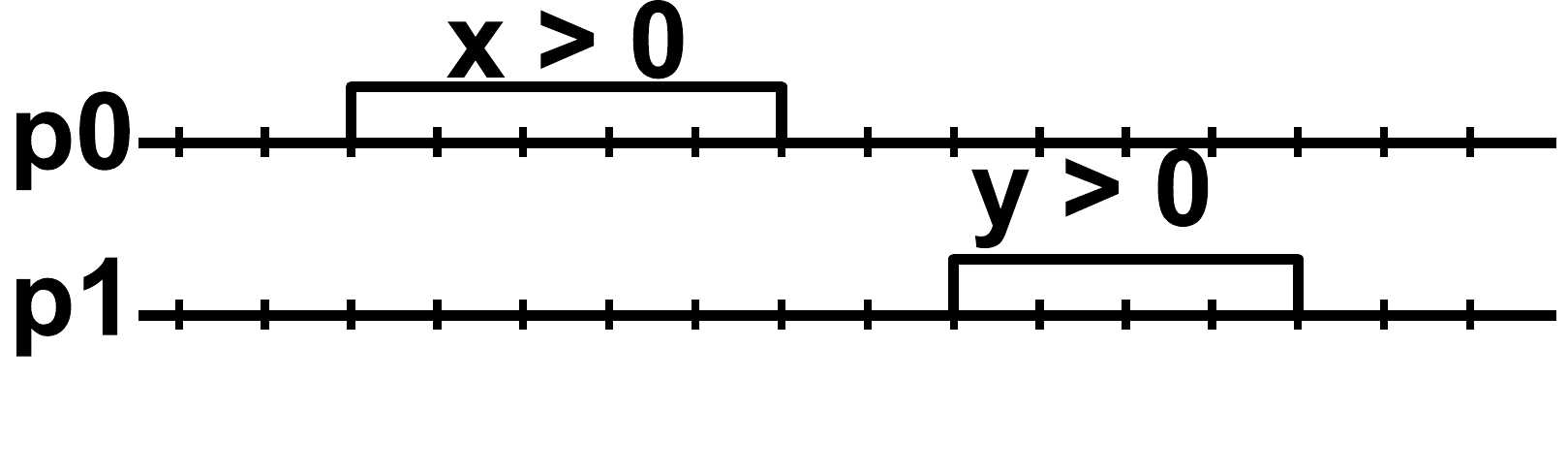}}
\end{center}
 \vspace{-10pt}
\vspace*{-4mm}

\caption{Uncertainty in Distributed Systems}
\label{fig:effectoftimeout}
\vspace{-0.2in}
\end{minipage}
\end{figure*}

%

Our goal in this work is to characterize the false positives/negatives in
run-time monitoring of a distributed system due to the uncertainty and impedance
mismatch. We focus on conjunctive predicates, i.e., predicates that are
conjunctions of local predicates of individual processes. The disjunction of
such conjunctive predicates can express any predicate in the system.
Our analysis focuses on comparing the application ground truth (whether the
predicate was true under the assumptions made by the application) with the
monitor ground truth (whether the predicate is true under the assumptions made
by the monitor). In other words, it identifies the effect of uncertainty in the
{\em problem} of monitoring distributed programs rather the uncertainty
associated with a given algorithm.

Specifically, we consider the following problems in the context of detecting
weak conjunctive predicates. (1) Suppose we utilize a monitoring algorithm
designed for asynchronous systems; then what is the likelihood of the result
being a false positive/negative when used with an application that relies on
partial clock synchronization. (2) Suppose we utilize an algorithm designed for
partially synchronous systems where it is assumed that clocks of two processes
are synchronized up to $\epsilon_{mon}$, but in reality, the bound used by
the application is $\epsilon_{app}$. In this context, what is the likelihood of
receiving false positive/negative detection? Moreover, if $\epsilon_{app}$
cannot be precisely identified (may have temporal perturbations), how
{\em sensitive} is the debugging algorithm to variations in clock drift/uncertainty?

\noindent{\em Precision, recall, and sensitivity of asynchronous monitors. } \
We present an analytical model that characterizes the false positive rate for
monitors that assume that the system is fully asynchronous (i.e.,
$\epsilon_{mon}=\infty$) and clock drift can be arbitrary ($\epsilon_{app}$ is
finite). Under these assumptions, monitor can only suffer from false positives:
The monitor will have perfect recall (i.e., there will be no false negatives)
but may suffer from a lack of precision. Our analytical results show that we
can classify the clock synchronization requirement in the partial synchrony
model into 3 categories with respect to two parameters \epone and \eptwo. We
find that if the clock drift is between $[0..\epone]$ then the precision of
monitoring is very low (i.e., the rate of false positives is high). If the drift
is in the range $[\eptwo .. \infty]$ then the precision of monitoring is
reasonably high. Moreover, in both of those cases, the precision is not very sensitive,
i.e., changes in the clock drift of the application does not affect the
rate of false positives.
However, in the range $[\epone ..\eptwo]$, the monitoring is hypersensitive and
small differences between the clock drift assumed by the monitor and the
underlying application can have a substantial impact on the rate of false
positives. A noteworthy result in this context is that the
hypersensitivity range $\frac{\eptwo-\epone}{\eptwo}$ approaches to 0
whenever $n\rightarrow \infty$.

\noindent{\em Precision, recall, and sensitivity of partially synchronous monitors. } \ %
We consider an extension of asynchronous monitors to the general case where the monitor relies on the fact that the underlying clocks are synchronized to be within $\epsilon_{mon}$, which may be different than the timing properties $\epsilon_{app}$ of the application. We find that for small $\epsilon_{app}$ there is a tradeoff among precision, recall, and sensitivity. If the monitor tries to achieve very high recall and precision (say at 95\%) at the same time, it becomes hypersensitive with respective to both precision and recall (small mismatch between the synchrony assumptions of the monitor and the underlying program can have a substantial impact on the rate of both false positives and false negatives). In this case, the monitor would need to sacrifice from the quality of either precision and recall to avoid being hypersensitive.
We also find that for large  $\epsilon_{app}$, the tradeoff dilutes. The monitor can achieve very high recall and precision while remaining less susceptible to sensitivity for large  $\epsilon_{app}$.

\noindent{\em Precision and recall of using fully synchronous monitoring for quasi-synchronous systems.} \ 
Finally we consider using fully synchronous monitors (i.e., $\epsilon_{mon} =0$) for monitoring quasi-synchronous systems. Quasi-synchronous systems are partially synchronous systems with the additional condition that if two events have the same physical clock value then they could have possibly happened at the same time. Adopting a quasi-synchronous model allows us to obviate the need for using vector clocks~\cite{fidge,mattern} and instead use inexpensive hybrid logical clocks~\cite{HLC} for predicate detection/monitoring. We investigate precision and recall tradeoffs to the face of clock drift/uncertainty in quasi-synchronous systems. For reasons of space, these results are provided in Appendix.

\vspace*{3mm}
\noindent{\em Implications of our findings for monitor design/tuning.} \ 
Our findings inform the monitor designer to manage the tradeoffs among
precision, recall, and sensitivity according to the predicate detection task at
hand. Our analytical model can inform based on $\epsilon_{app}$ and local
predicate occurrence probability, whether hypersensitivity is avoidable or not.
If hypersensitivity is avoidable, $\epsilon_{mon}$ can be chosen from the
suitable interval to achieve both high precision and high recall. However, if it
becomes necessary to make a tradeoff between precision and recall to avoid
hypersensitivity, the monitor would need to decide which one is more important,
and which one it can sacrifice.

The monitor may decide to prioritize recall in lieu of reduced
precision. In other words, the monitor can attain better coverage of
notifications of predicate detection to the expense of increased false positive
notifications. This is useful for investigating predicates that occur rarely,
where one can't afford to miss occurrences of the predicate but can afford to
investigate/debug some false-positive detections. This is also useful for
monitoring safety predicates, which is relatively easier to debug.

The monitor may decide to prioritize precision in lieu of reduced
recall. In other words, the monitor can reduce the false positive notifications
of predicate detection to the expense of allowing some missed notifications of
predicate detection. This is useful for predicates that occur frequently: the
monitor has enough opportunities to sample and can afford to miss some
occurrences of the predicate. This is also useful for monitoring
liveness/progress predicates, which is harder to debug and false-positives cause
wasting time with debugging.

{\bf Organization of the paper. } \
In Section \ref{sec:systemmodel}, we present our computational model. 
In Section \ref{sec:vc-effectiveness}, we investigate precision and sensitivity of asynchronous monitors in partially synchronous systems. 
In Section \ref{sec:hvc-eps-effectiveness}, we analyze the precision, recall, and sensitivity of partially synchronous monitoring of partially synchronous systems. 
We discuss related work in Section \ref{sec:related} and conclude in Section \ref{sec:concl}. 
Finally, in Appendix, we provide proofs of theorems in the paper and discuss additional results for monitoring quasi-synchronous systems.

\section{System Model} 
\label{sec:systemmodel}
\newcommand{\pred}{\ensuremath{\mathcal{P}}\xspace}

We consider a system that consists of a set of $n$ processes that communicate via messages. Each process has a local clock
that is synchronized to be within $\epsilon$ of absolute time, using a protocol such as NTP~\cite{ntp}. Any message sent in the system is received no earlier than $\delta_{min}$ time and no later than $\delta_{max}$ time. 
We denote such a system as $\systype{\epsilon,\delta_{min}, \delta_{max}}$-system. 
We also use the abbreviated notion of $\systype{\epsilon,\delta }$-system, where $\delta$ denotes the {\bf minimum} message delay and the maximum message delay is $\infty$. 
Observe that this modeling is generic enough to model asynchronous systems $(\epsilon = \infty, \delta_{min} = 0, \delta_{max} = \infty)$ and purely synchronous systems $(\epsilon = 0, \delta_{min} = 0, \delta_{max} = 0)$, as well as partially synchronous systems.

We define $\hb-consistent$ to capture the requirement that two events $e$ and $f$ ``could have'' happened at the same time. Specifically, $e$ and $f$ are $\hb$-consistent (also called concurrent) provided both $e \hb$ $f$ and $f \hb$ $e$ are false.
\footnote{Following Lamport's definition of causality~\cite{lamport}, for any
  two events $e$ and $f$, we say that $e \hb$ $f$ ($e$ happened before $f$) iff
  (1) $e$ and $f$ are events in the same process and $e$ occurred before $f$,
  (2) $e$ is a send event and $f$ is the corresponding receive event, and (3)
  there exists an event $g$ such that $e \hb$ $g$ and $g \hb$ $f$.}
If both $e \hb$ $f$ and $f \hb$ $e$ are false then $e$ and $f$ could have happened at the same time in an asynchronous system where clock drift could be arbitrary. A global snapshot consisting of local snapshot of each process is $\hb$-consistent iff all local snapshots are mutually $\hb$-consistent. 

For partially synchronous systems, we define the notion of $\epsilon$-consistent. Two events $e$ and $f$ are $\epsilon$-consistent provided they are $\hb$-consistent and the difference between the physical time of $e$ and $f$ is no more than $\epsilon$. 
A global snapshot consisting of local snapshot of each process is $\epsilon$-consistent iff all local snapshots are mutually $\epsilon$-consistent. 

A conjunctive predicate $\pred$ is defined of the form $\pred_1 \land \pred_2
\land \dots \land \pred_n $, where $\pred_i$ is a local predicate at process
$i$. At each process, the local predicate $\pred_i$ can be randomly and
independently truthified at the chosen time unit granularity (say millisecond
granularity) with probability $\beta$. For instance, if $\beta=0.1$ and time
unit is selected as millisecond, then the local predicate is truthified roughly
every 10 milliseconds.
We use $\ell$ to denote the length of an interval for which the local predicate remains true at a process once it is truthified. 

\section{Precision and Sensitivity of Asynchronous Monitors} 
\label{sec:vc-effectiveness}

In this section, we evaluate the precision and sensitivity of an asynchronous monitor in partially synchronous systems. In particular, we focus on $\systype{\epsilon, \delta}$ systems.




\ifsinglecolumn
\subsection{Analytical Model} 

Using a monitor designed for asynchronous systems in partially synchronous systems can result in a false positive. Hence, in this section, we develop am analytical model to address the following question:
\begin{quote}
If we use a monitor for predicate detection that is designed for an asynchronous system and apply it in a partially synchronous system, what is the likelihood that it would result in a false positive?
\end{quote}

The false positive rate is defined as the probability of a snapshot discovered by the asynchronous monitor  is a false positive in $\systype{\epsilon,\delta}$-system. To compute this ratio for interval-based local predicates, we first define the followings. Two intervals $[a_1,b_1]$ and $[a_2,b_2]$ differ by $\max(\max(a_1,a_2)-\min(b_1,b_2),0)$. Let $c$ be a snapshot consisting of a collection of intervals $[a_i,b_i]$ for each process $i=0$ to $n-1$. We denote $L(c)$ as a length of snapshot defined by the least value of $x$ such that $c$ is $x$-consistent snapshot.

If an $hb$-consistent snapshot is also $\epsilon$-consistent, this is a true positive, which means the asynchronous monitor is precise in this case.
Our first result is Precision (true positive rate) of $hb$-consistent snapshots in $\systype{\epsilon,\delta}$-systems. For reasons of space, the proof is relegated to the Appendix.

\begin{thm}
For interval-based predicate, given $c$ is an $hb$-consistent snapshot, the probability of $c$  being also $\epsilon$-consistent is  $$ \phi( \epsilon,n,\beta, \ell) = (1-(1-\beta)^{\epsilon+\ell-1})^{n-1} $$
\label{thm:vcinterval}
\end{thm} \vspace{-0.2in}

%
The formula above suggests that when $n$ increases, snapshots that are \emph{hb}-consistent will become less physically consistent. 
This is expected since the more number of processes, the harder to find \emph{hb}-consistent snapshots as well as physically close \emph{hb}-consistent snapshots.
On the contrary, if we increase $\beta$, predicates will be more frequent and there are more physically close \emph{hb}-consisten snapshots.

We use the characteristics of this function to compute the sensitivity of asynchronous monitors.  We focus on $\epsilon$ since it is likely to vary over time. We consider the special case where predicates are point-based denoted as $\phi( x,n,\beta) = \phi( x,n,\beta, 1)$.

We identify two inflection points of $\frac{\partial \phi( \epsilon,n,\beta)}{\partial \epsilon}$, denoted as $\epsilon_{p_1}$ and $\epsilon_{p_2}$ where  $\phi$ changes rapidly for $\epsilon \in [\epsilon_{p_1},\epsilon_{p_2}]$. On the other hand, we observe that if $\epsilon \leq \epsilon_{p_1}$ or $\epsilon \geq \epsilon_{p_2}$, the change of $\phi$ is very small. That is, in the range, $[0 .. \epsilon_{p_1}]$, the monitor has lots of false-positives and is not very sensitive to changes in the value of $\epsilon$. Moreover, in range $[\epsilon_{p_2}, \infty]$, the monitor has little false positives and again not sensitive to changes in the value of $\epsilon$. However, in the range $[\epsilon_{p_1} .. \epsilon_{p_2}]$, the monitor is very sensitive to changes in $\epsilon$. In other words, except in the range $[\epsilon_{p_1} .. \epsilon_{p_2}]$, we can compute the precision of the asynchronous monitor even with only an approximate knowledge of $\epsilon$ used in the partially synchrony model. 

Our next result shows that the gap between two inflection points of $\frac{\partial \phi( \epsilon,n,\beta)}{\partial \epsilon}$ approaches zero for large $n$.

\begin{thm}
For $n > 1$, two inflection points of $\frac{\partial \phi( \epsilon,n,\beta)}{\partial \epsilon}$ are at $$ \{\epsilon_{p1},\epsilon_{p2}\} = \log_{(1-\beta)}(\frac{3n-4 \pm \sqrt{5n^2-16n+12}}{2(n-1)^2}) $$ 
 Where $\epsilon_{p1} < \epsilon_{p2}$. Furthermore, the relative \uncertain range approaches $0$ as $n$ increases. In other words, the relative difference of phase transition $\epsilon_{p_1}$ and post-phase transition $\epsilon_{p_2}$ converges to 0, which is independent of $\beta$. That is, $$\lim_{n\rightarrow \infty} \frac{\epsilon_{p_2}-\epsilon_{p_1}}{\epsilon_{p_1}} = 0$$

\label{thm:uncertainrange}
\end{thm}

This result means there is an abrupt change (i.e., phase transition) between the range of lots of false postives to the range of little false positives.  For small value of $n$ such as 100, the value of $\frac{\epsilon_{p_2}-\epsilon_{p_1}}{\epsilon_{p_1}}$ is approximately 0.52. The absolute value of $\epsilon_{p1}$ depends on $n,\beta$. 

To understand the main point of the result,
we can instantiate some concrete values. For example, taking the unit of time granularity as millisecond, with $n = 50$ and $\beta = 0.001$ (i.e., the local predicate is true every second on average), the two points of inflection of slope are at 3635.41 and 5550.24 msecs respectively. This means if the system has $\epsilon$ less than 3 seconds, then with high probability the $hb$-consistent global conjunctive predicate is not $\epsilon$-consistent. If the system has  $\epsilon$ more than 6 seconds, then with high probability the $hb$-consistent detection is also $\epsilon$-snapshot. 
As another example, with  $n = 50$ and $\beta =  0.5$, the two points of inflection of slope are at 5.24 and 8.01 msecs, respectively. This means with high probability $hb$-consistent predication is $\epsilon$-consistent if the system has $\epsilon$ greater than 8 msecs.

\fi

\subsection{Simulation setup}
\label{sec:simulationmodel}


To validate the analytical model, we set up a simulation environment. The simulation code is available at \url{http://www.cse.msu.edu/~nguye476/}.
In our simulation, at any given instance, with a certain probability a process chooses to advance its clock as long as the synchrony requirement will not be violated. When a process increments its clock, it can decide if the local predicate is true with probability $\beta$. Depending upon point-based detection and/or interval-based detection, the local predicate will remain true for just one instant or for a duration whose length is chosen by an exponential distribution.
Furthermore, when a process advances its clock, it can choose to send a message to a randomly selected process with probability $\alpha$. The delay of this message will be $\delta$, the minimum message delay. Note that the analytical model predicts that the possibility that a given cut is a false positive is independent of $\alpha$ and $\delta$. We find that this result is also valid with simulations. Hence, delivering the message as soon as it is allowed does not change the false positive rate. The values of $\alpha$ and $\delta$ only affects the number of snapshots identified.

%

%
\begin{wrapfigure}{r}{0.50\textwidth}
\vspace{-30pt}
\begin{center}
\includegraphics[width=0.45\textwidth]{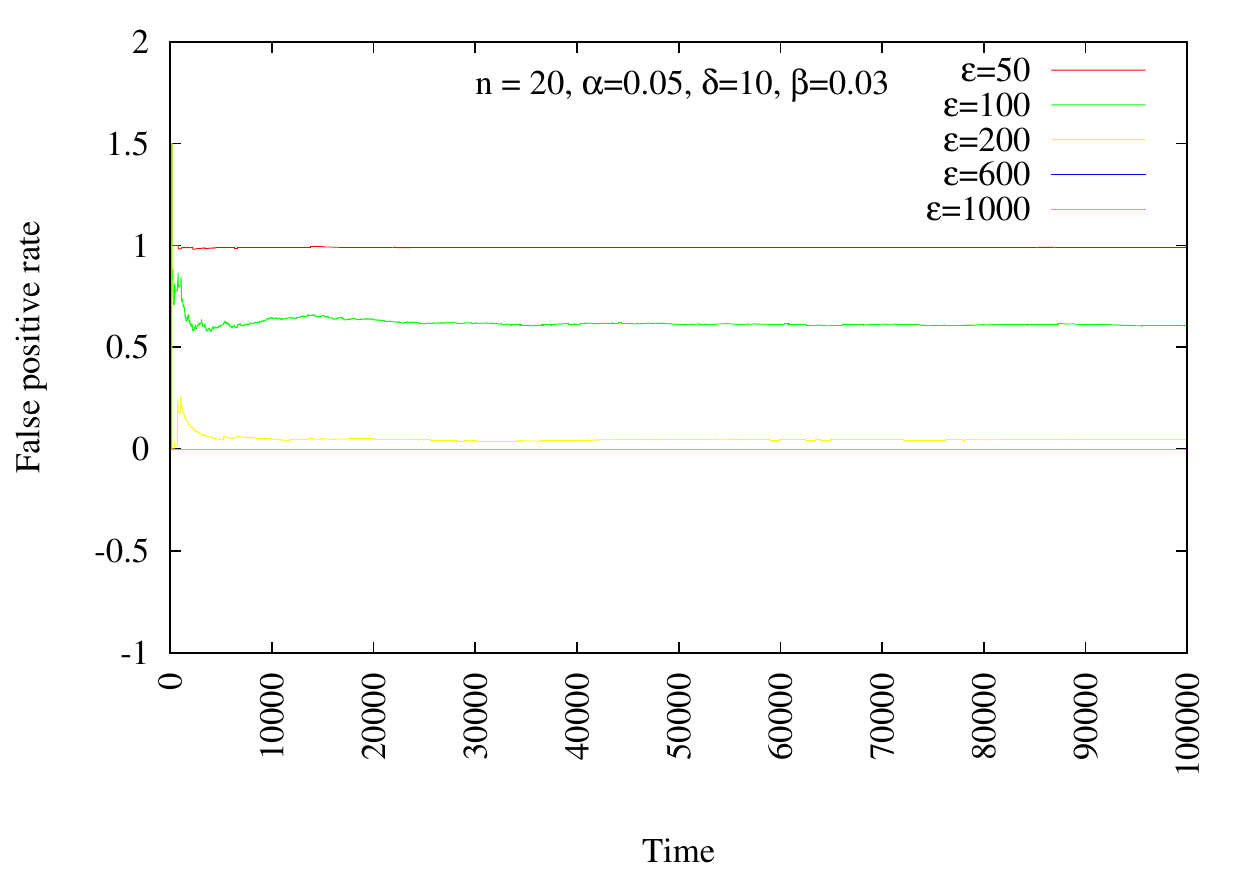}
\end{center}
\vspace{-20pt}
\caption{Convergence of false positive rates}
\label{fig:fprtime}
\vspace{-15pt}
\end{wrapfigure}
We first validate that running the simulation upto time $100,000$ computes
stable false positive rates.
Initially the false positive rate varies substantially.
However, when sufficiently many snapshots are identified it stabilizes to the
desired value. To validate this, we considered how the false positive rates vary
during different simulations. Figure \ref{fig:fprtime} show these results for
different values of $\epsilon$ while $n$, $\delta$, $\alpha$ and $\beta$ are
fixed. From these results, we find that the desired false-positive ratio
stabilizes fairly quickly. When we vary $n$, $\delta$, $\alpha$ and $\beta$, we
also observe a similar stabilizing pattern.

%

Thus, we run our simulations until each process advances its clock to $100,000$. During this run, we identify $Y$, the number of snapshots identified by the asynchronous monitor algorithm in~\cite{garg},  and $Y_F$, the number of snapshots that are also $\epsilon$-consistent. Thus, the false positive rate $FPR$ is calculated as $1-  \frac{Y_{F}}{Y}$.

\subsection{Sensitivity for point-based predicates}
\label{sec:basic}


{\bf Independence of false positive rate with respect to $\alpha$ and $\delta$. } \
\begin{wrapfigure}{r}{0.45\textwidth}
 \vspace{-30pt}
\begin{center}
\includegraphics[width=0.45\textwidth]{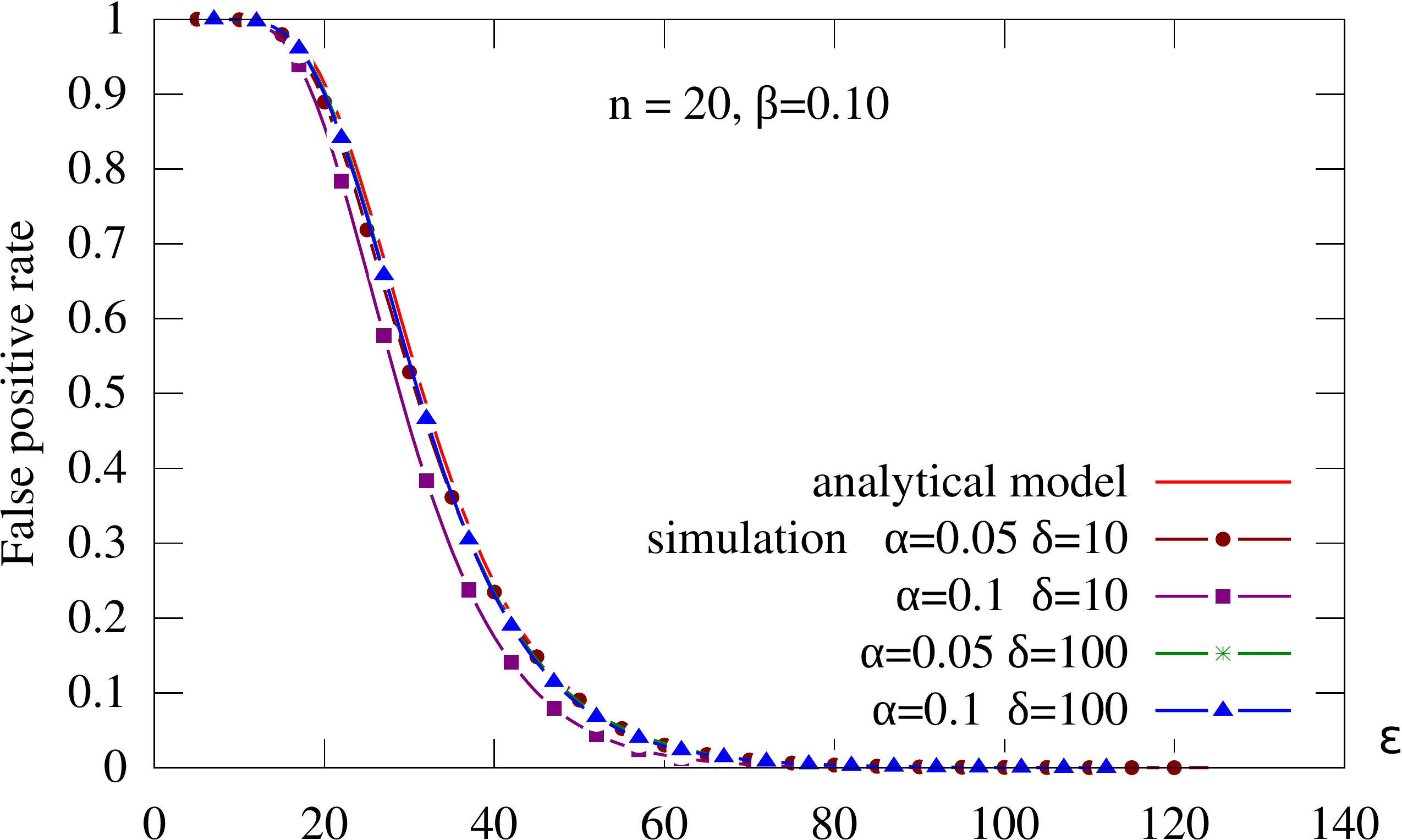}
\end{center}
 \vspace{-20pt}
\caption{The independence of false positive rate with $\alpha$ and $\delta$, shown by analytical model and simulations}
\label{fig:fpr-alpha-delta}
 \vspace{-25pt}
\end{wrapfigure}
Since the analytical model predicts that the false positive rate is independent of $\alpha$ and $\delta$, we validate this result with our simulation. Specifically, Figure \ref{fig:fpr-alpha-delta} considers the false positive rates for $n=20, \beta=0.10$. We consider different values of $\alpha=0.05$, $0.1$ and $\delta=10$, $100$ and compare the simulation results with the analytical model. The simulation results validate the analytically computed false positive rate as well as the fact that it is independent of $\alpha$ and $\delta$.
\\
{\bf Effect of $\epsilon$. } \
Figures \ref{fig:fpr-vs-eps}(a)-\ref{fig:fpr-vs-eps}(b) illustrate the effect of false positive rate for different values of $\epsilon$. Figure \ref{fig:fpr-vs-eps}(a), \ref{fig:fpr-vs-eps}(b) consider the cases with $n=5$, and $20$ processes, respectively. In each figure, we vary $\beta$ from $1\%$ to $8\%$. The results validate the analytical model's prediction that values of $\epsilon$ can be divided into 3 ranges: a brief range of lots of false positives to the left when $\epsilon$ is small, a range of little false positives to the right when $\epsilon$ is large, and a short uncertainty range in the middle where small change in $\epsilon$ significantly effects the false positive rates.

{\bf Effect of $\beta$. } \
As expected from the analytical model, when the value of $\beta$ is close to $0$, the predicted false positive rate is $1$. And, as $\beta$ approaches $1$, false positive rate approaches $0$. We validate this result with Figure \ref{fig:fpr-vs-eps}(b). When considering a network of $20$ processes, and $\beta$ is small, say $1\%$, the false positive rate at $\epsilon=200$ is $93.51\%$. By contrast if $\beta$ is increased to $3\%$ and $5\%$ then the false positive rate decreases to $4.68\%$ and $0.08\%$ respectively.

\begin{figure*}[tbhp]
\begin{minipage}{\linewidth}
\vspace{-15pt}
\begin{center}
\subfigure[]{\includegraphics[width=0.46\columnwidth]{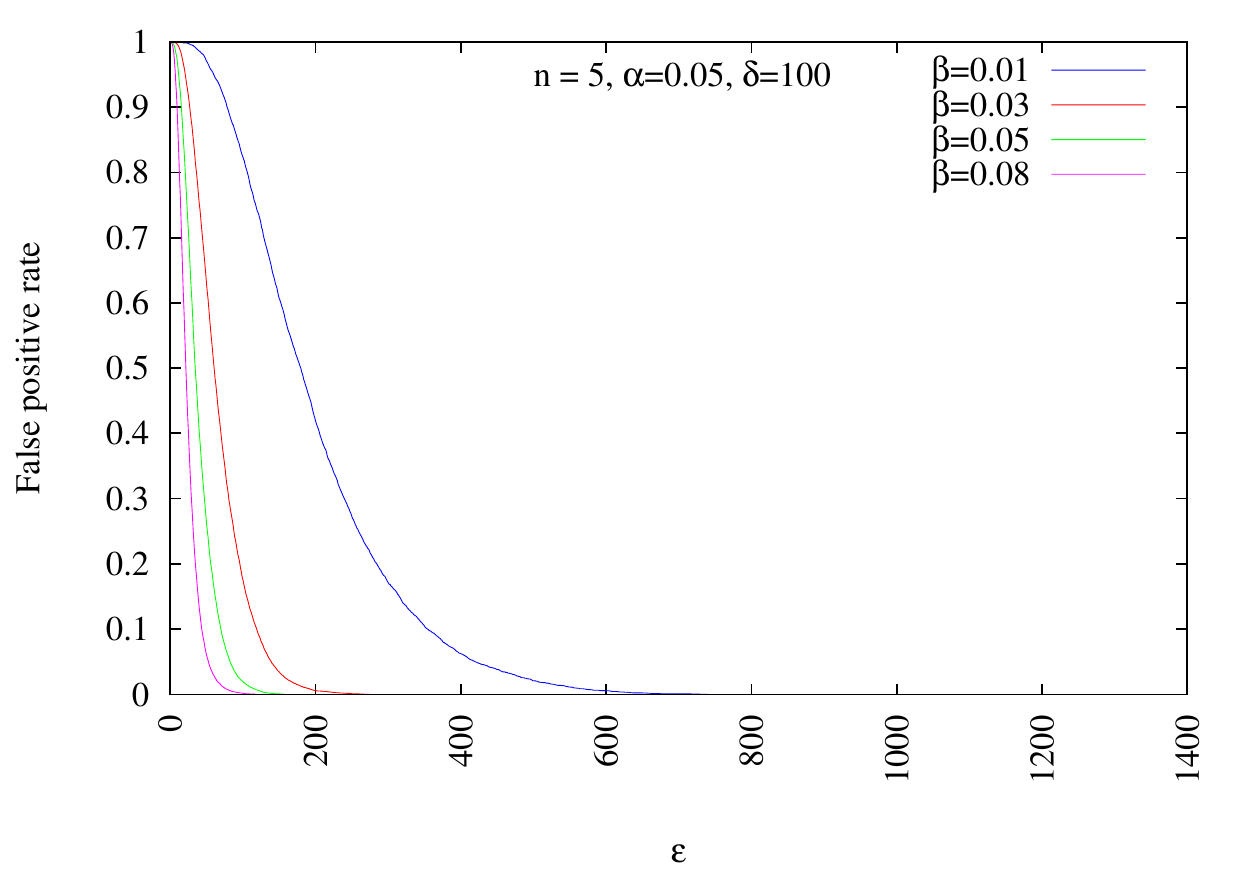}}
\subfigure[]{\includegraphics[width=0.46\columnwidth]{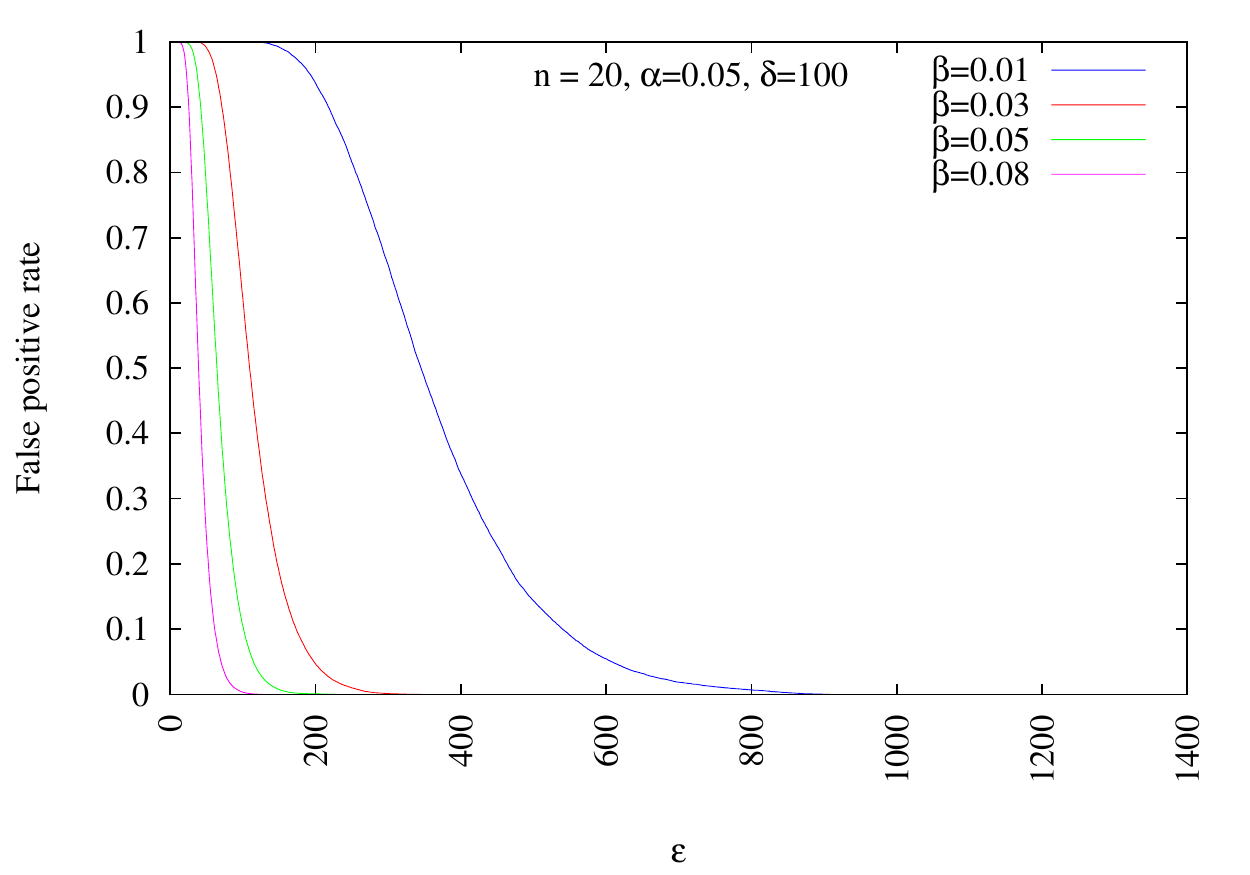}}
\end{center}
\vspace{-15pt}
\caption{Impact of $\beta$ and $n$ on false positive rates.}
\label{fig:fpr-vs-eps}
\vspace{-10pt}
\end{minipage}
\end{figure*}

{\bf Effect of $n$. } \
The analytical model predicts that when $n$ increases, the false positive rate increases. The speed of change depends on particular $\beta$. This result is confirmed in Figures \ref{fig:fpr-vs-eps} (a), \ref{fig:fpr-vs-eps}(b). Let $\beta=0.01$, when $n$ is small, say $5$, the false positive rate at $\epsilon=200$ is $43.71\%$. If $n$ is increased to $20$ then the false positive rate increases to $93.51\%$.

\subsection{Sensitivity for correlated point-based predicates}
\label{sec:vc-correlation}


In our analytical model, we assumed that $\beta$, the probability of local predicates being true, is independent at each process. In this section, we compare the analytical model with simulations where the truth value of local predicates on different processes is correlated. While we consider some specific approaches to add correlation below, our analysis technique is useful for several other correlations as well. 

We first consider a model abbreviated as \emph{PMA} (Positively correlated with MAjority). In \emph{PMA} correlation model, processes are divided into 2 groups. Each process in the first group of size $G_1$ generates predicates independently with the same base rate $\beta$ at each clock tick. A process in the second group of size $G_2 = n - G_1$ has 2 possibilities: (1) either it follows the majority of the first group with probability $P_{dep}$ or (2) chooses the truth value independently of others with rate $\beta$ (with probability $P_{ind} = 1 - P_{dep}$). The values of $G_1, G_2, P_{dep}, P_{ind}$ are configurable.

In a rough estimation of the false positive rate under \emph{PMA} model, we observe that given that the local predicates in group $G_1$ are close enough in a snapshot, the chance for the snapshot to be a false positive would depend on whether the predicates in group $G_2$ are close enough to the first group or not. This would in turn depend on cases where predicates in group $G_2$ are independently generated (if they are dependently generated, they would be close the predicates of the first group).

The probability that a predicate independently generated by a process in group $G_2$ is at time $t$ apart from group $G_1$ follows a geometric distribution which is $(1 - \beta_{ind})^{t-1}\beta_{ind}$  where $\beta_{ind} = P_{ind}*\beta$. Hence, the probability that all $G_2$ processes of the second group are within the $\epsilon$ distance from the first group is roughly $\phi_{PMA}(\epsilon, n, \beta) = 1 - (1 - (1-\beta_{ind})^\epsilon)^{G_2}$.

As shown in Figure \ref{fig:hpma}, 
\ref{fig:thnma}, these estimates agree well with our simulation results when we vary the values of $P_{ind}$ and $G_2$.

\begin{figure*}
    \vspace{-10pt}
    \begin{center}
        \subfigure[]{%
            \label{fig:hpma}
            \includegraphics[height=0.2\textwidth, width=0.45\textwidth]{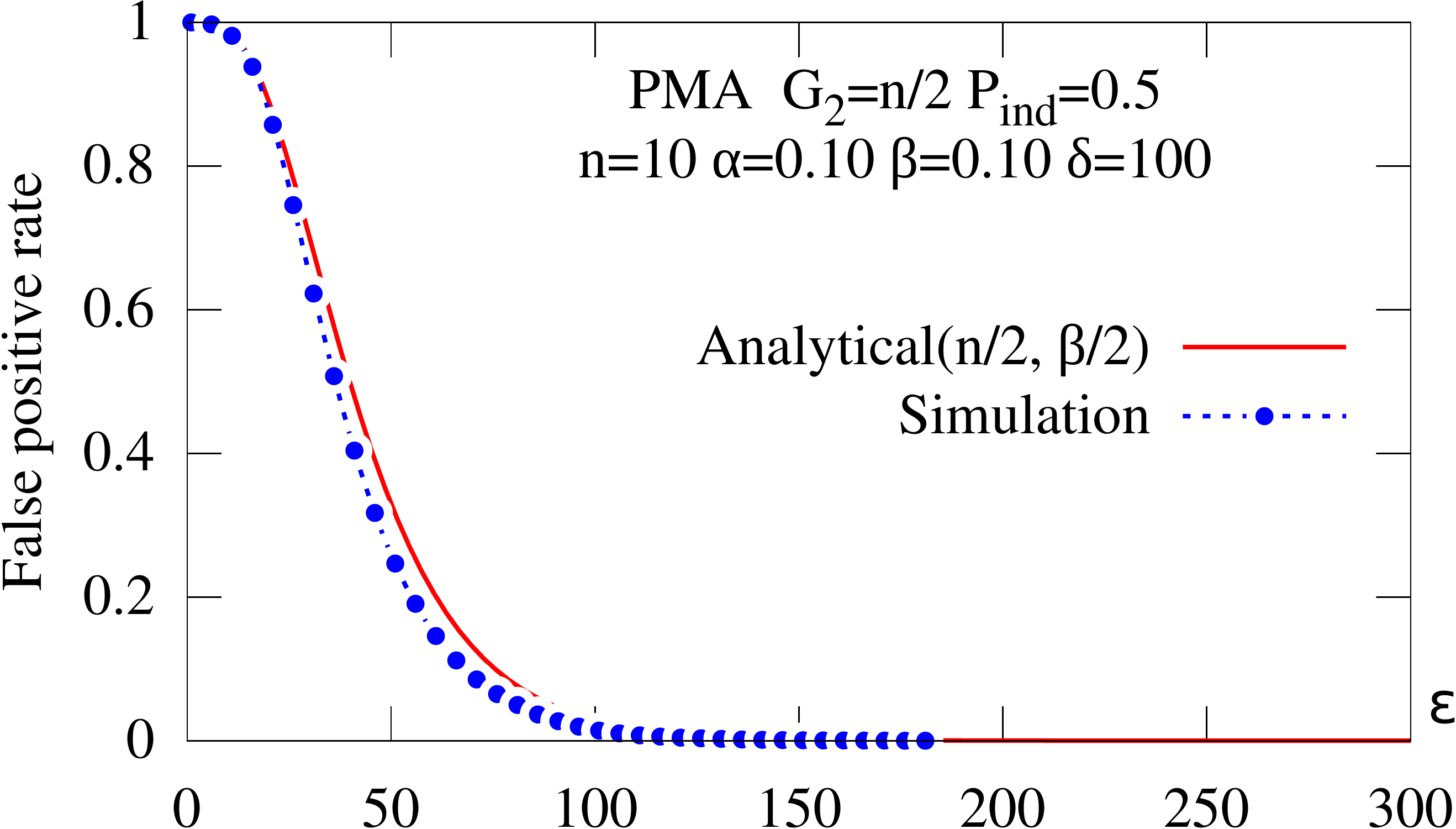}
        }%
        \subfigure[]{%
           \label{fig:thnma}
           \includegraphics[height=0.2\textwidth, width=0.45\textwidth]{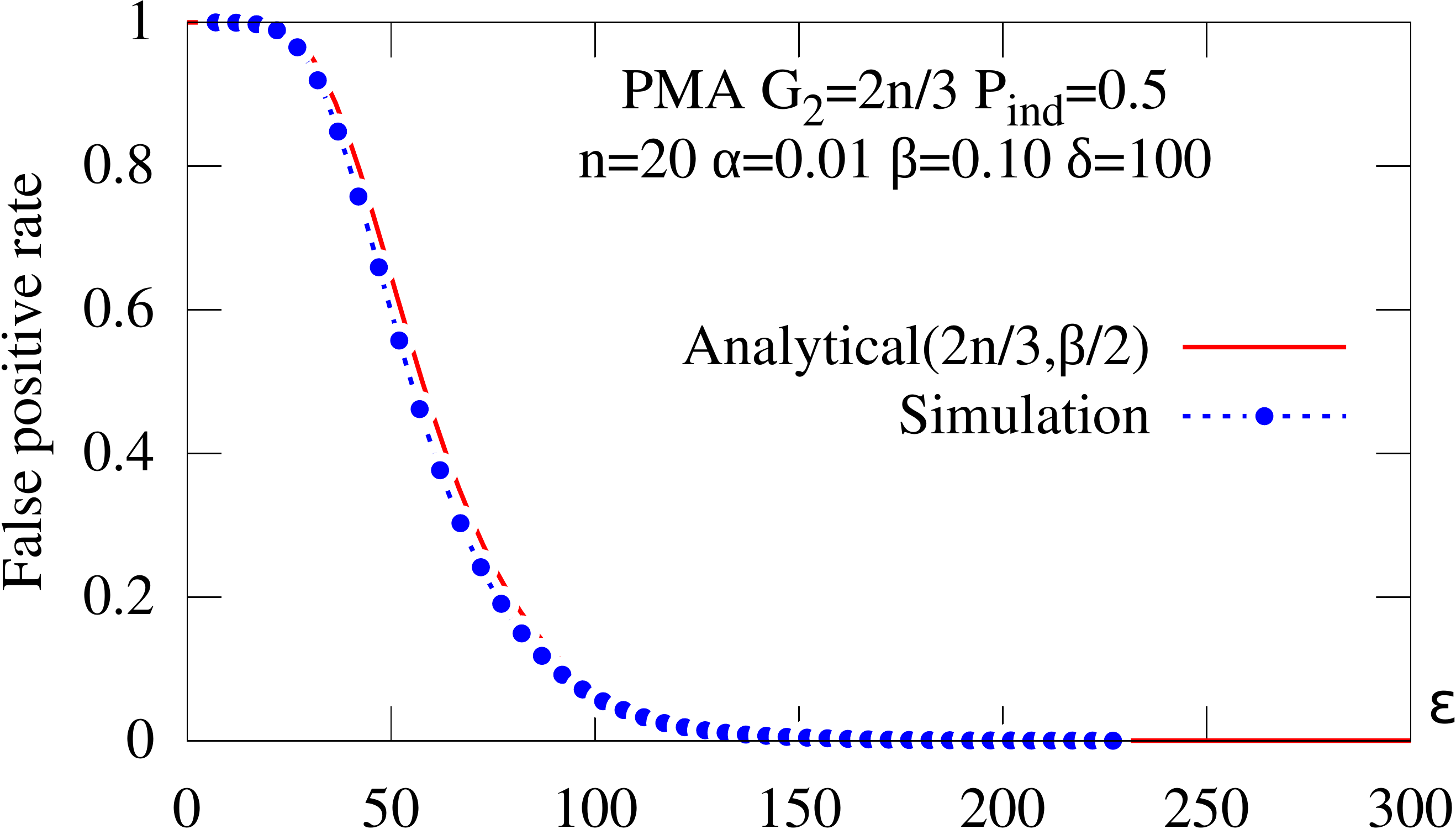}
        } \\
        \vspace{-5pt}
        \subfigure[]{%
            \label{fig:hnma}
            \includegraphics[height=0.2\textwidth,width=0.45\textwidth]{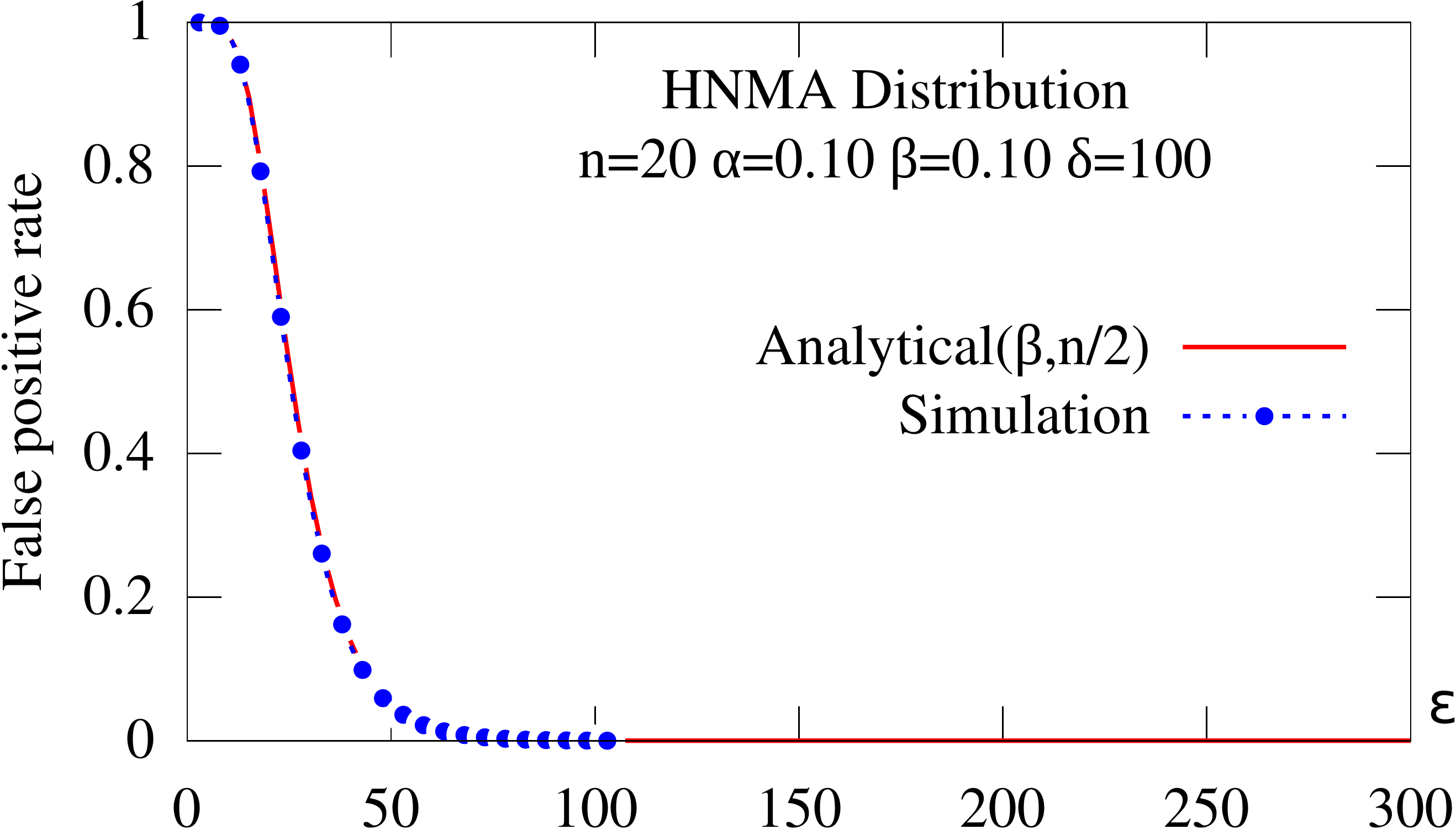}
        } 
        \subfigure[]{%
           \label{fig:pmaj20}
           \includegraphics[height=0.2\textwidth, width=0.45\textwidth]{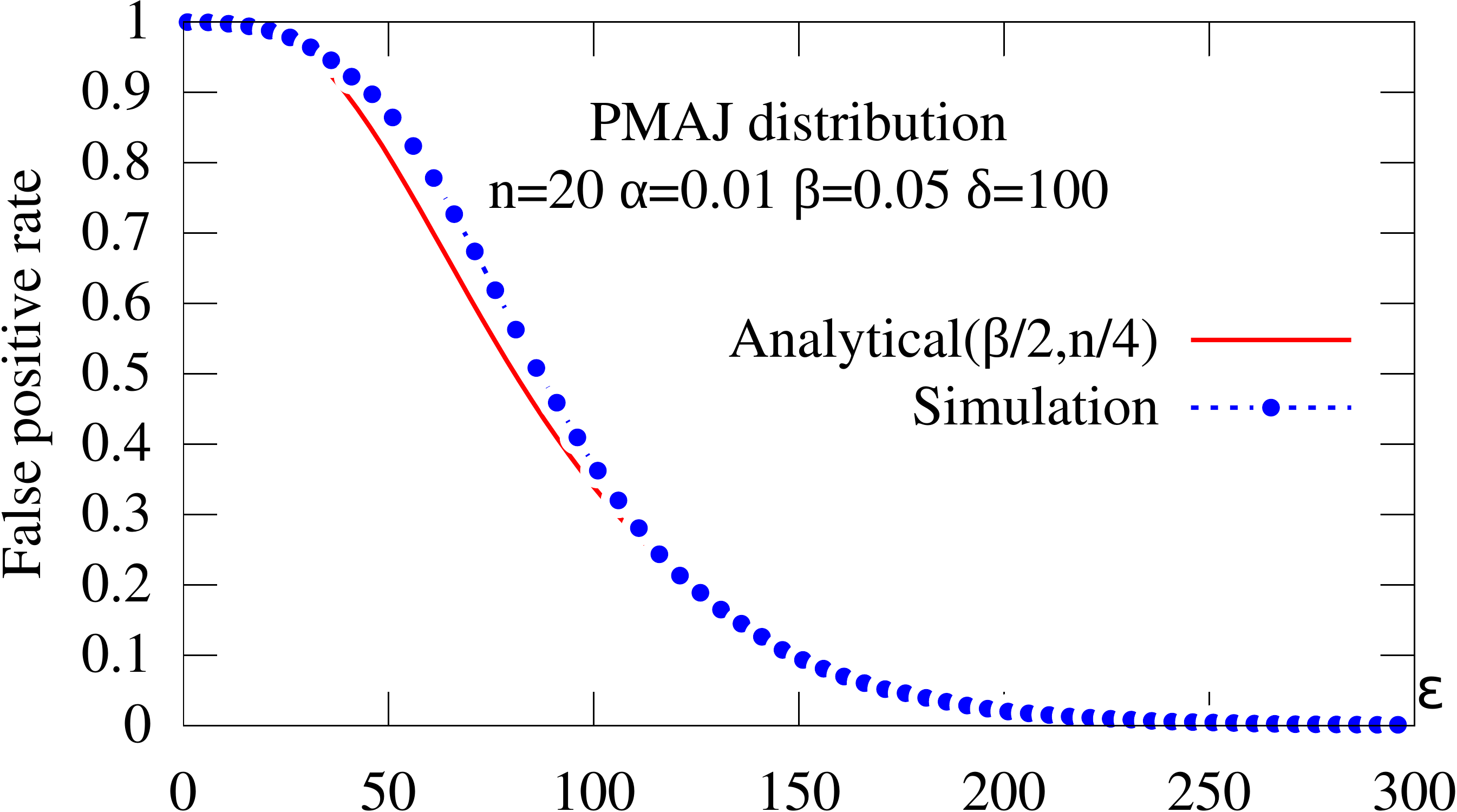}     
        } \\
        \vspace{-5pt}
    \end{center}
    \vspace{-10pt}
    \caption{%
        False positive rates in some correlated models.
     }%
    \label{fig:correlation}

\end{figure*}


We also consider other correlation models for processes' predicates such as \emph{HNMA} and \emph{PMAJ}.
The \emph{HNMA} (Half Negatively correlated with MAjority) model is the similar to \emph{PMA} where $G_1 = \frac{n}{2}, P_{dep} = 0.5$  with one exception: processes in the second group would follow the minority
of the first group.
In \emph{PMAJ} (Positively correlated with MAjority upto index J) model, process $0$ chooses whether its predicate is true with probability $\beta$. The truth value of other predicates is correlated with predicates in its preceding processes (w.r.t. process ID).
In particular, each process $j$ will follow the majority of its preceding processes (i.e. processes $0, ..., j-1$)  with probability of $0.5$; for another $0.5$ probability it will generate predicates by its own $\beta$. %
As shown in Figure \ref{fig:hnma}, 
\ref{fig:pmaj20}, there are parameters that help our analytical model to estimate the
simulation results of these correlation models under different parameter settings (e.g. $n, \beta$). For example, the effective $(n,\beta)$ for the correlation
model \emph{HNMA} and \emph{PMAJ} are $(\frac{n}{2},\beta)$, $(\frac{n}{4},\frac{\beta}{2})$ respectively.

The interprocess correlation of truth value of local predicates could perturb the curve of the false positive rate by pulling or pushing the curve. We have considered several such correlations and find that the analytical model matches the simulation model reasonably closely if we update a new value of $n$ and a new value of $\beta$. 

\subsection{Sensitivity with interval-based predicates}
\label{sec:interval}



Point-based scenarios could be generalized to interval-based scenarios where local predicates is true for a certain interval of time, $\ell$.
Interval based predicates are also more likely in partially synchronous model where the local predicate is expected to be true for a certain duration.

In this case, the false positive rate needs to computed slightly differently. In particular, suppose we have an $hb$-consistent snapshot that is not $\epsilon$-consistent, it may still be possible there is another snapshot that contains the same set of messages (in terms of send and receive) and is $\epsilon$-consistent. In particular, such an $\epsilon$-consistent message would simply delay certain processes to obtain an $\epsilon$-consistent snapshot. (As an example, consider Figure \ref{fig:interval-example}. the snapshot consisting of $e_1$ and $f_1$ is not $\epsilon$-consistent if $\epsilon=10$. But $e_2$ and $f_1$ is $\epsilon$-consistent.)
Also, we do not identify two similar snapshots as distinct snapshots. In particular, in Figure \ref{fig:interval-example},  we compute $e_3$ and $f_2$ as the same snapshot as $e_2$ and $f_1$.

\begin{figure}[tbhp]
 \vspace{-10pt}
\centering
\includegraphics[width=0.5\textwidth]{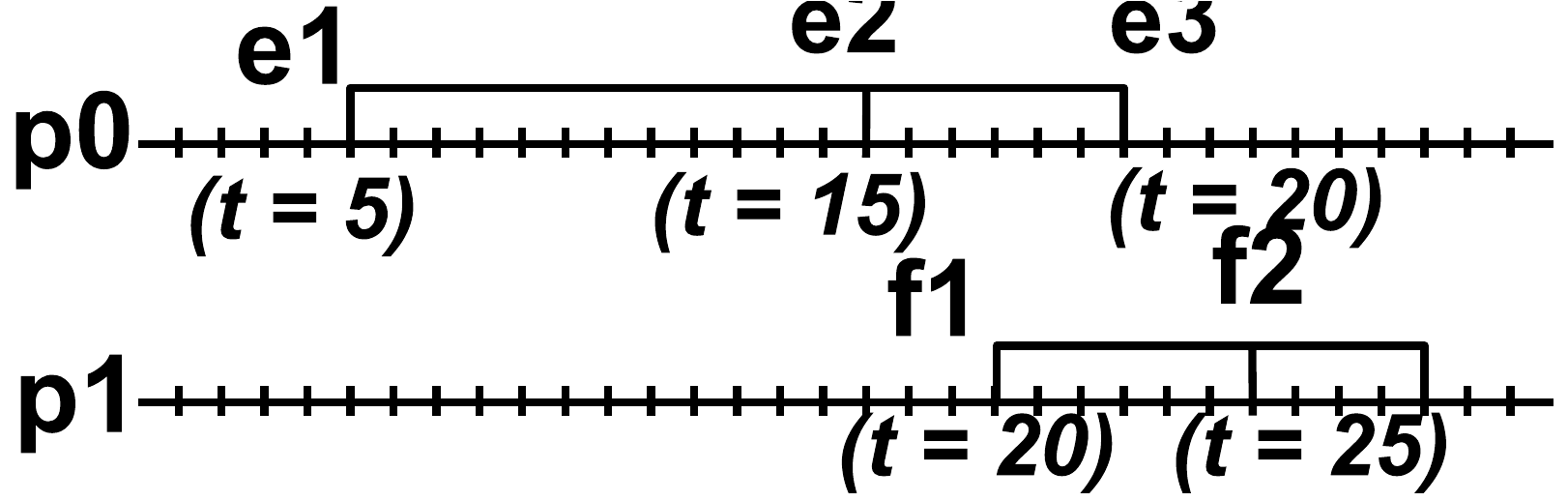}
\caption{Interval-based predicate detection.}
\vspace{-10pt}
\label{fig:interval-example}
\end{figure}

In our simulations, when the local predicate is marked to be true, it also identifies
%
\begin{wrapfigure}{r}{0.50\textwidth}
\vspace{-20pt}
\begin{center}
\includegraphics[width=0.45\textwidth]{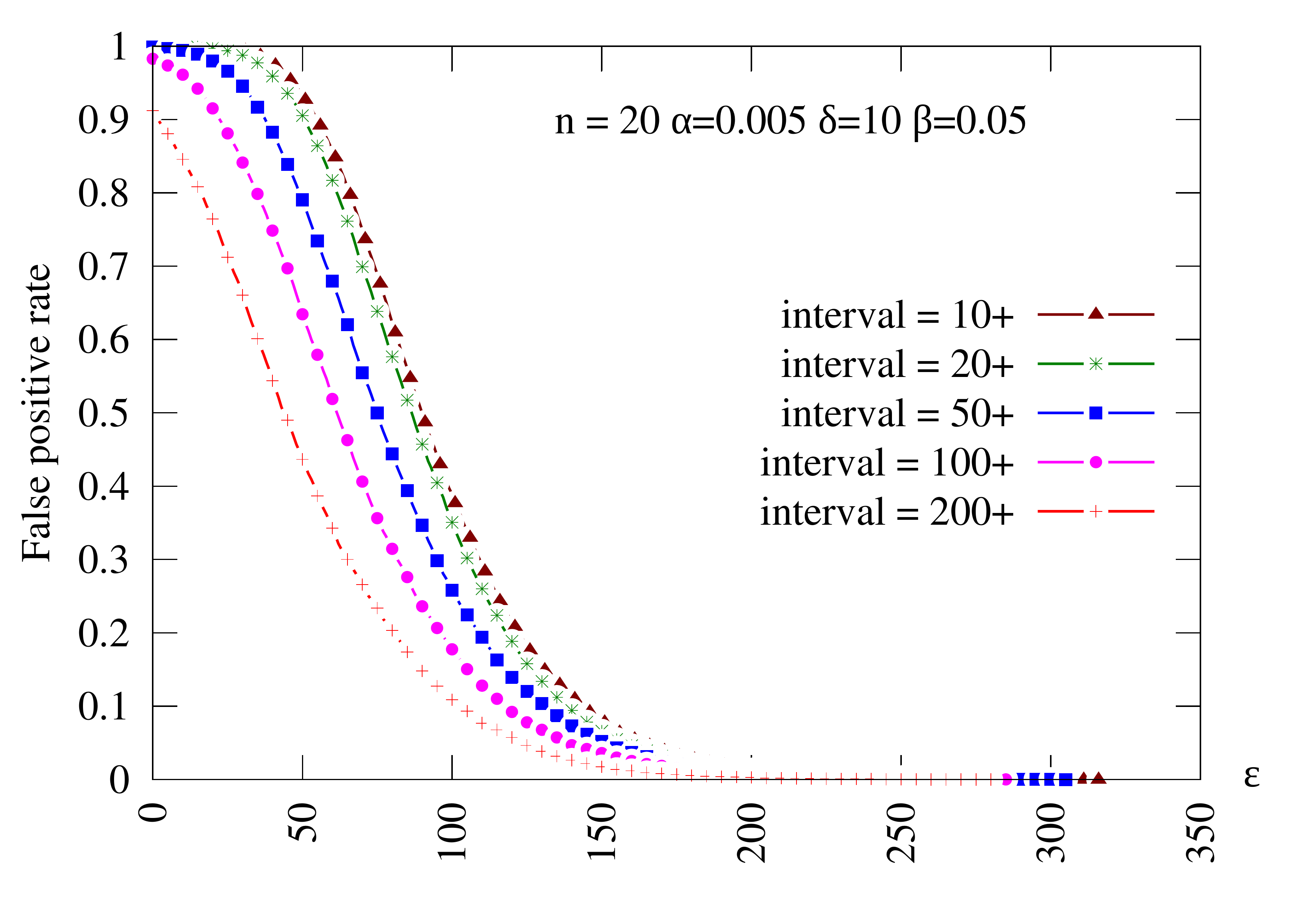}
\end{center}
\vspace{-15pt}
\caption{Impact of interval length on interval-based predicate detection.}
\label{fig:interval_size_epsilon}
\vspace{-20pt}
\end{wrapfigure}
an interval for which the predicate will remain true. The length of the interval is selected randomly from a geometric distribution with probability of success $0.3$ (a randomly chosen value for these experiments).
%
%
Other than this, there is no change in the simulation model compared to that considered in Section \ref{sec:simulationmodel}.


%

Simulation results in Fig. \ref{fig:interval_size_epsilon} show that the false positive rate for interval-based detection is similar to that of point-based cases. Secondly, as the interval size gets larger, the false positive rate becomes smaller. These results are compatible with the analytical model.


%

\section{Precision, Recall, and Sensivity of Partially Synchronous Monitors} 
\label{sec:hvc-eps-effectiveness}

In this section,
we focus on the following problem:


\begin{quote}
Suppose we designed a monitor (predicate detection algorithm) for a $\systype{\epsilon_{mon}, \delta_1 }$-system and applied it in a system that turns out to be a $\systype{\epsilon_{app}, \delta_2 }$-system, then what are possible false positives/negatives that may occur?  \footnote{As validated in Section \ref{sec:vc-effectiveness}, the value of $\delta$ is not important. Hence, we only focus on the relation between $\epsilon_{mon}$ and $\epsilon_{app}$.} 
\end{quote}

\subsection{Analytical Model and its Validation with Simulation Results}
\label{sec:hvc-analytical}

We consider the case where the monitoring algorithm assumes partially synchronous model where clocks do not differ by more than $\epsilon_{mon}$. This algorithm is then used for monitoring an application that implicitly relies on the assumption that clocks are synchronized to be within $\epsilon_{app}$, that is difficult to compute and is unavailable to the monitoring algorithm. 
Such an application may make use $\epsilon_{app}$ with the use of timeouts, or even more implicitly by relying on database update and cache invalidation schemes to ensure that no two events that are more than $\epsilon_{app}$ can be part of the same global state as observed by the clients~\cite{faceConsistency}.

If $\epsilon_{app} < \epsilon_{mon}$, then the situation is similar that of the asynchronous monitors, where $\epsilon_{mon} = \infty$. However, if $\epsilon_{mon}$ is finite then it will reduce the false positives as this monitor will avoid detecting some instances where the time difference between the local predicates being true is too large. Thus, a monitor that assumes that clocks are synchronized to be within $\epsilon_{mon}$, will detect snapshots that are $\epsilon_{mon}$-consistent. However, it was only supposed to identify $\epsilon_{app}$-consistent snapshots. Hence, the precision of the algorithm, i.e., the ratio of the number of snapshots correctly detected and number of snapshots detected, can be determined by calculating the probability that an $\epsilon_{mon}$-consistent snapshot is also an $\epsilon_{app}$-snapshot. Also, in this case since every $\epsilon_{app}$-consistent snapshot is also an $\epsilon_{mon}$-consistent snapshot, the monitor will recall all correct snapshots. 

If $\epsilon_{app} > \epsilon_{mon}$, the situation would be reversed, i.e., precision will always be 1. But recall would be less than 1, as the monitor may fail to find some snapshots that are $\epsilon_{app}$ consistent but not $\epsilon_{mon}$-consistent. Thus, we have




\begin{thm}
When a monitor designed for \systype{\epsilon_{mon}, \delta}-system is used in an application that assumes that the system is \systype{\epsilon_{app}, \delta}-system, the Precision and Recall are as follows:

\begin{tabbing} 

$Precision =  \frac{f(\min(\epsilon_{app},\epsilon_{mon}))}{f(\epsilon_{mon})}$,  \hspace*{5mm} \= False positive rate = $1 - Precision $\\
$ Recall = \frac{f(\min(\epsilon_{app},\epsilon_{mon}))}{f(\epsilon_{app})}$ \> False negative rate = $1 - Recall$\\
Where $ f(x) = (1-(1-\beta)^{x+\ell-1})^{n-1} $

\end{tabbing}
\vspace{-0.1in}
\label{thm:hvcpr}
\end{thm}

Next, we study the \textit{sensitivity} --changes in the value of Precision and Recall 
based on changes in $|\epsilon_{app} - \epsilon_{mon}|$-- of partially synchronous monitor. 
We visualize this by a diagram called PR-sensitivity Diagram using Precision and Recall. PR-sensitivity Diagram is basically a contour map of Precision and Recall given two variables $(\epsilon_{mon},\epsilon_{app})$. If $\epsilon_{app} > \epsilon_{mon}$, the diagram shows only Recall since Precision in this area is always one. Similarly, if $\epsilon_{app} < \epsilon_{mon}$, the diagram shows only Precision.   Let $\acbound$ be an accuracy bound, meaning that Precision and Recall are bounded by $\acbound$, PR-sensitivity Diagram shows contour whose value is $\acbound$. 

\begin{wrapfigure}{r}{0.50\textwidth}
\vspace{-20pt}
\begin{center}
\includegraphics[width=0.45\textwidth]{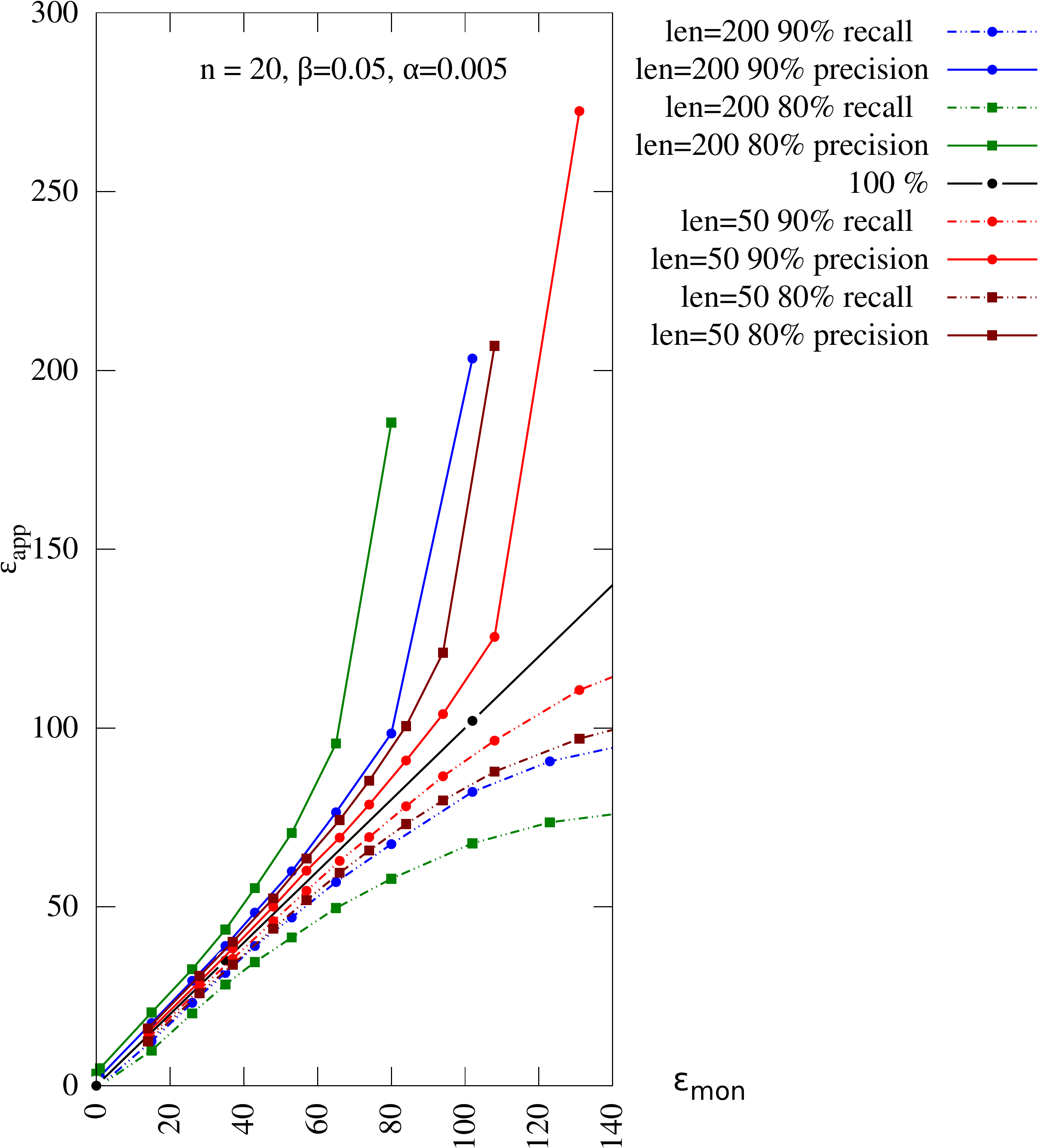}
\end{center}
\vspace{-10pt}
\caption{Precision and Recall when varying interval size}
\label{fig:prec-recall-int-vary-length}
 
\end{wrapfigure}
Figures \ref{fig:prec-recall-point-based}(a) and \ref{fig:prec-recall-point-based}(b) show examples of PR-sensitivity Diagram. This diagram shows that the contour lines of Precision/Recall move closer as $\epsilon_{app}$ gets small. In other words, \textit{the value of Precision and Recall is sensitive when $\epsilon_{app}$ is small}. If $\epsilon_{mon} > \epsilon_{app}$ (respectively, $\epsilon_{mon} < \epsilon_{app}$), then even minute change in $\epsilon_{app}$ can result in large change in Precision (respectively, Recall). In this case, we need to be careful when monitoring in such tight synchronization. For scenarios where we consider intervals where local predicates are true, the results are shown in Figure \ref{fig:prec-recall-int-vary-length}. As anticipated, the longer the intervals, the better precision and recall. 

We describe analytical result. If we want both Precision and Recall to be greater than $\acbound$, the relation between $\epsilon_{mon}$ and $\epsilon_{app}$ needs to satisfy the condition in the next theorem. Observe that this theorem identifies useful range --where both precision and recall are greater than $\acbound$-- of a monitor. The proof is in Appendix. 

\begin{figure*}[!tbhp]
\begin{minipage}{\linewidth}
 \vspace{-20pt}
\begin{center}
\subfigure[]{\includegraphics[width=0.45\columnwidth]{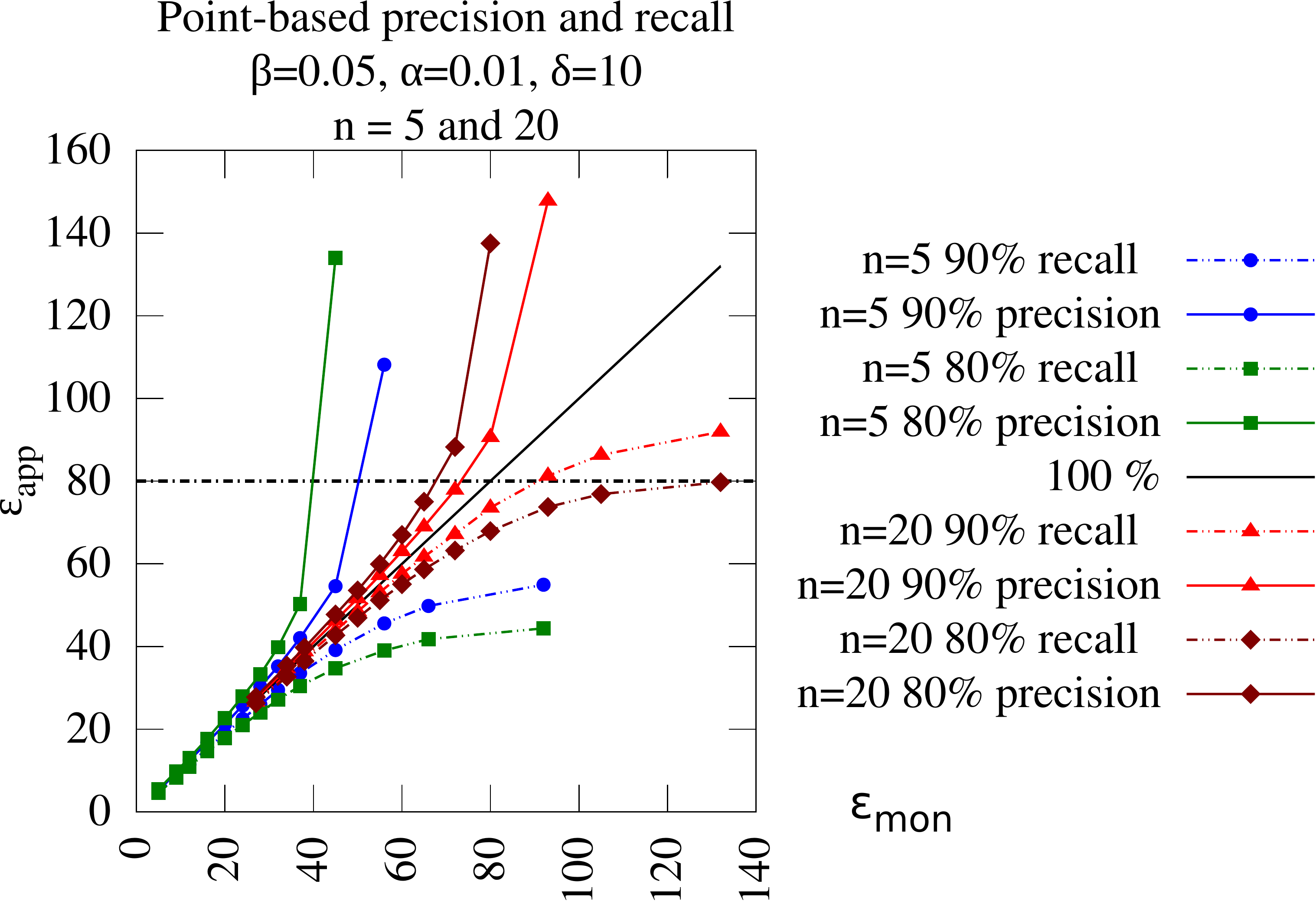}}
\subfigure[]{\includegraphics[width=0.45\columnwidth]{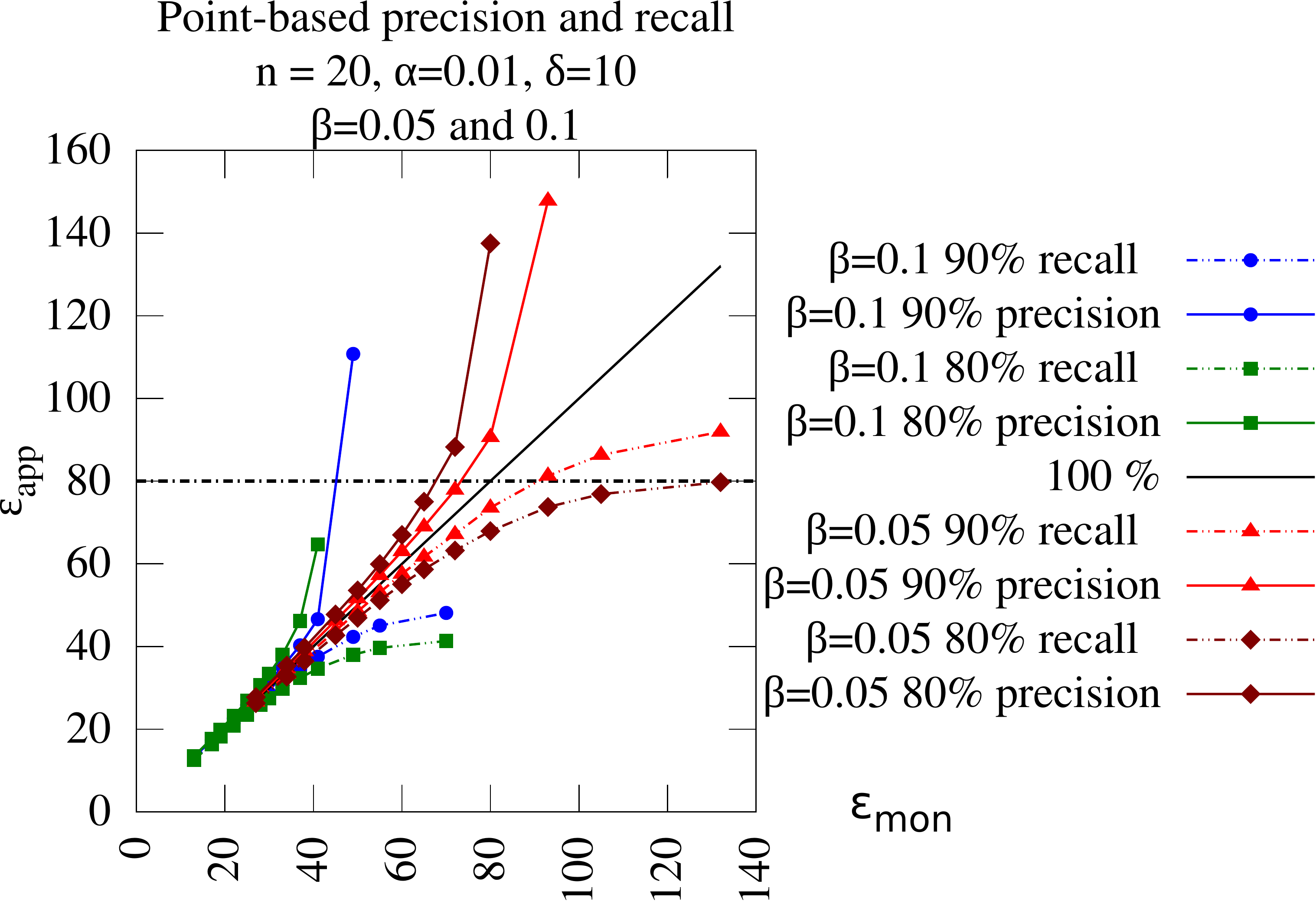}}
\end{center}
 \vspace{-10pt}
\caption{Precision and Recall Diagram in point-based predicate detection.} 
\label{fig:prec-recall-point-based}
\vspace{-15pt}
\end{minipage}
\end{figure*}

\begin{thm}
 For \systype{\epsilon_{app},\delta}-system with $n$ processes where each process has probability $\beta$ to have predicate true, the monitor designed for $\systype{\epsilon_{app}, \delta}$ system has Precision and Recall no less than $\acbound$ if the following condition holds:
$$ \log_{1-\beta}(1-\eta^{\frac{1}{n-1}} g(\beta,\epsilon_{app},\ell)) \leq \epsilon_{mon}+\ell -1 \leq \log_{1-\beta}(1-\eta^{\frac{-1}{n-1}}g(\beta,\epsilon_{app},\ell))$$

Where $$ g(\beta,\epsilon_{app},\ell) = 1-(1-\beta)^{\epsilon_{app}+\ell-1} $$
\label{thm:hvcaccurate}
\end{thm}

Finally, there is a phase transition such that if $\epsilon_{app}$ is too small then the precision and recall are hypersensitive, meaning that a minute change can result in drastically different accuracy. If $\epsilon_{app}$ is beyond phase transition, then the precision and recall are almost non-sensitive as the bound in Theorem \ref{thm:hvcaccurate} grows rapidly. The proof is in Appendix. 

\begin{thm}
 The Precision and Recall due to difference in $\epsilon_{app}$ and $\epsilon_{mon}$ is \textit{hypersensitive} if and only if $$ \epsilon_{app} \leq \log_{1-\beta}(\eta^{\frac{-1}{n-1}}-1) - \ell +1 $$ 

\label{thm:hvcphasetransition}
\end{thm}

\begin{wrapfigure}{r}{0.45\textwidth}
 \vspace{-30pt}
\begin{center}
\includegraphics[width=0.45\textwidth]{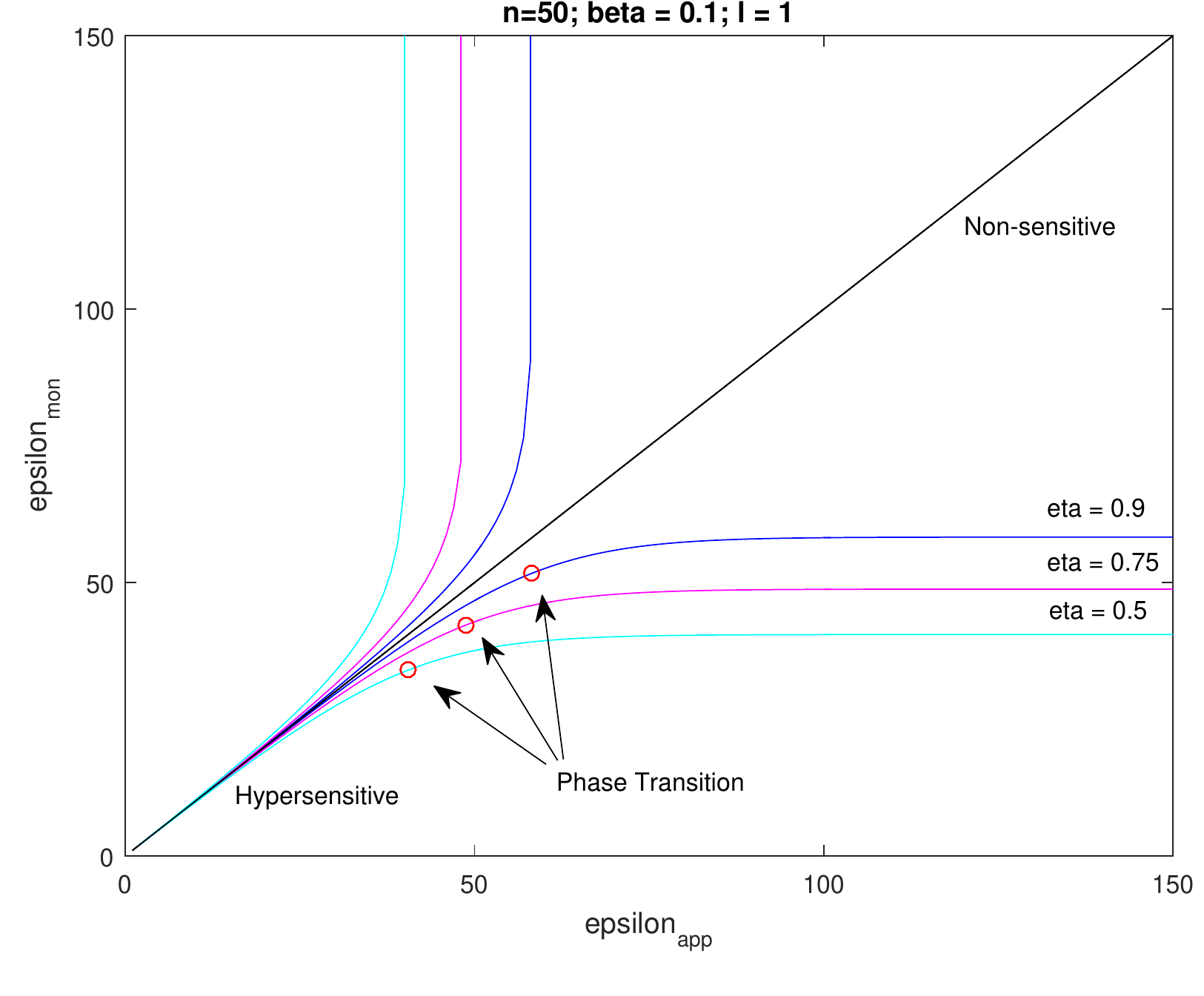}
\end{center}
 \vspace{-20pt}
\caption{PR diagram given by analytical model.}
\label{fig:analyticalPR}
 \vspace{-20pt}
\end{wrapfigure}
Illustration of Theorem \ref{thm:hvcaccurate} and \ref{thm:hvcphasetransition} is shown in Figure \ref{fig:analyticalPR}. Suppose there are 50 processes where each local predicate truthification rate is at every 10ms (thus, $\beta = 0.1$). The bounds are obtained by Theorem \ref{thm:hvcaccurate}. Each red circle highlight the point of maximum concavity, which is the starting point of phase-transition as highlighted in Theorem \ref{thm:hvcphasetransition}. Notice that after phase transition for each value of $\eta$, there is virtually no sensitivity at all as we can deviate from $\epsilon_{app}$ while maintaining high precision and recall. However, if $\epsilon_{mon}$ is less than phase transition, the regions below are hypersensitive. In this case, we cannot obtain both high precision and recall simultaneously. Instead, we can choose to have high precision while sacrificing recall and vice versa.

\section{Related Work} 
\label{sec:related}

Inherent to the model of shared nothing distributed systems is that the nodes execute with limited information about other nodes. This further implies that the system developers/operators also have limited visibility and information about the  system. Monitoring/tracing and predicate detection tools are an important component of large-scale distributed systems as they
provide valuable information to the developers/operators about their system under deployment.

\noindent {\bf Monitoring large-scale web-services and cloud computing systems.} \ 
Dapper~\cite{dapper} is Google's production distributed systems tracing infrastructure. The primary application for Dapper is performance monitoring to identify the sources of latency tails at scale.
Making the system scalable and reducing performance overhead was facilitated by the use of adaptive sampling. The Dapper team found that a sample of just one out of thousands of requests provides sufficient information for many common uses of the tracing data.

Facebook's Mystery Machine~\cite{mysteryMachine} has similar goals to Google's Dapper.
Both also use similar methods, however mystery machine tries to accomplish the task relying on less instrumentation than Google Dapper. The novelty of the mystery machine work is that it tries to infer the component call graph implicitly via mining the logs, where as Google Dapper instrumented each call in a meticulous manner and explicitly obtained the entire call graph.

\noindent {\bf Predicate detection with vector clocks.} \ There has been a lot of previous work on predicate detection (e.g., Marzullo \& Neiger~\cite{marzulloDet} WDAG 1991, Verissimo~\cite{verissimoDet} 1993), using vector clock (VC) timestamped events sorted via happened-before (hb) relationship.
The work in~\cite{marzulloDet} not only defined Definitely and Possibly detection modalities, but also provided algorithms for predicate detection using VC for these modalities. That work also showed that information about clock synchronization (i.e., $\epsilon$) can be translated into additional happened-before constraints and fed in to the predicate detection algorithm to take into account system synchronization behavior and avoiding false positives in only VC-based predicate detection. However, that work did not investigate the rates of false-positives with respect to clock synchronization quality and event occurance rates.

\noindent {\bf Predicate detection with physical clocks and NTP synchronization. } \ In partially synchronized systems, Stoller~\cite{stollerDet} investigated global
predicate detection using NTP clocks, showing that using NTP
synchronized physical clocks provide some benefits over using VC in terms of complexity of predicate detection. The worst case complexity for predicated
detection using hb captured by VC is $\Omega(E^N)$, where $E$ is the maximum
number of events executed by each process, and $N$ is the number of processes.
With some assumptions on the inter-event spacing being larger than time synchronization uncertainty, it is possible to have
worst-case time complexity for physical clock based predicate detection to be
$O(3^N E N^2)$ --- linear in $E$. 

\noindent {\bf Predicate detection under partially synchronous system.}
The duality of the literature on monitoring predicates forces one to make a binary choice before hand: Go with either VC- or physical clock-based timestamping and detection \cite{HLC,dk13ladis}. Hybrid Vector Clocks (HVC) obviate this duality and offer the lowest cost detection of VC and physical clock-based detection at any point. Moreover while VC is of $\Theta(N)$ \cite{vcsize}, thanks to loosely-synchronized clock assumption, it is possible with HVC to keep the sizes of HVC to be a couple entries at each process \cite{HVC}. HVC captures the communications in the timestamps and provides the best of VC and physical clock worlds.

\noindent {\bf Runtime monitoring with imprecise timestamp. } \
Prior runtime-verification approaches assume timestamp to be precise.  However, results from such protocol may not be correct due to uncertainty in underlying system. Recent works account for clocks' imprecision under variety settings. Zhang et al. \cite{Zhang2010} proposes probabilistic approach to deal with imprecise timestamp in data-stream processing. Wang et al. \cite{Wang2011} consider imprecise trace in runtime verification due to  unknown event ordering.  Basin et al. \cite{Basin2014} focus on the real-time temporal logic MTL over a continuous time domain that accounts for imprecise timestamp.  Implicitly, those assumption can be too strong as well. Our result shed light on how sensitive of errors from assumptions we made in the system need to be so that the overall error rate is acceptable. 

\section{Conclusion} 
\label{sec:concl}

We presented analytical and simulation models to capture the effect of the gap between assumptions made by the application and by the monitor. First, we investigated the effect of using a monitor designed for asynchronous systems in partially synchronous systems.
%
%
We find that regarding $\epsilon$, we can partition the system in three regions: lots of fault positive, uncertain range, little false positives.  We find that the uncertain range is hypersensitive, i.e., small changes in $\epsilon$ change the false positive rate substantially. We also showed how these ranges can be computed analytically. In particular, we show how one can compute $\epsilon_{p_1}$ and $\epsilon_{p_2}$ such that the lots of false positive range is $[0..\epsilon_{p_1}]$, \uncertain range is $[\epsilon_{p_1}..\epsilon_{p_2}]$ and little false positive range is $[\epsilon_{p_2}..\infty]$. 
An interesting observation in this context was that the uncertainty range, $\frac{\epsilon_{p_2}-\epsilon_{p_1}}{\epsilon_{p_2}}$, approaches $0$ as the number of processes increase or as $\epsilon_{app}$ grows. 
We also showed that although the analytical results focused on situations where the probability of the local predicate being true is independent, it can also be used in cases where local predicate being true is correlated.

We also considered the case where monitoring algorithm assumes that the clocks are synchronized to be within $\epsilon_{mon}$, but the actual clock synchronization of the system/program is $\epsilon_{app}$. One reason this may happen is that application uses clock estimation approaches to identify dynamic value of $\epsilon_{app}$ but this value is not visible to the monitor and, hence, it uses an estimated value. 
We identified possible ranges where the error rate caused by differences in these values is within acceptable limits. Here, we find that for specific ranges of $\epsilon_{mon}$, the algorithm is highly sensitive.  We observed tradeoffs among precision, recall, and sensitivity when  $\epsilon_{app}$ is small, and found that the tradeoff dilutes as $\epsilon_{app}$ gets larger.

There are several future extensions of these results. One extension is to evaluate error probability for more complex predicates in terms of conjunctive predicate detection. Here, if the predicate were $\phi_1\vee \phi_2$ there is a possibility that even if $\phi_1$ is detected incorrectly, $\phi_2$ may still be true causing detection of $\phi_1 \vee \phi_2$. Another extension is to address other types of predicates (e.g., $\phi_1$ leads-to $\phi_2$). 
In our work, we compared the false positives/negatives of monitors that detect a predicate iff it is true under the assumptions made by the monitor. Another future extension is to consider the case for specific instances of monitors which have potential in-built errors introduced for sake of efficiency during monitoring. 

\section{Acknowledgments}
This work is in part sponsored by the National Science Foundation (NSF) project under award number XPS-1533870.
 
\newpage

\bibliographystyle{IEEEtran}
\bibliography{murat,vidhya}

\newpage
\appendix
\noindent\textbf{\Large Appendix} 
\\
\section{Effectiveness of Monitors for Quasi-Synchronous Systems } 
\label{sec:hlc-effectiveness}
In Sections \ref{sec:vc-effectiveness} and \ref{sec:hvc-eps-effectiveness}, we consider the effectiveness of monitors designed for asynchronous system and partially synchronous system. The analysis in Section \ref{sec:hvc-eps-effectiveness} can be instantiated for the case where the system is fully synchronous, i.e., where clock drift is $0$. Although achieving fully synchronous clocks is difficult/impossible in a distributed system, they offer an inherent advantage. Specifically, in asynchronous/partially synchronous systems, to identify whether two events could have happened at the same time, we need to use techniques such as vector clocks \cite{fidge,mattern} that require $O(n)$ space where $n$ is the number of processes. Even though there are attempts to reduce the size \cite{HVC,BoundedVersionVectors}, the worst case size is still $O(n)$. By contrast, in fully synchronous systems, if two events happen at the same time on two different processes, we can conclude that they happened at the same time. In other words, $O(1)$ information suffices with fully synchronous clocks. 

Although fully synchronous {\em physical} clocks are hard to achieve, we can get {\em simulated} clocks that achieve the same property. Our goal in this work is to evaluate effectiveness of such an algorithm in monitoring partially synchronous systems. We denote such systems as quasi-synchronous systems. 

If the underlying system is fully synchronous, we can implement a monitoring algorithm as follows: If all local predicates are true at the same time $t$ then the conjunction of that predicate is true.\footnote{Note that our analysis is based on the property of the monitor and, hence, we do not consider how this can be evaluated most efficiently}. However, if the underlying application is using asynchronous/partially synchronous model then there may be errors. 

Observe that if the underlying system is partially synchronous but the monitor uses the above approach with simulated clocks that provide the desired property, it will suffer from false negatives. In other words, it may miss to find instances where the conjunctive predicate is true. This may happen if the events on two processes are within the uncertainty of clocks but do not happen exactly at the same time. 

With this motivation, we focus on the following problem: Given a $\systype{\epsilon, \delta}$ system that provides a simulated clock that guarantees that two events with equal simulated clock value are concurrent (i.e., do not depend upon each other) what is the rate of false negatives if one monitors using such simulated clocks. 
Since this analysis depends upon how the simulated clock is implemented (although not on how the monitoring itself is implemented given the simulated clock), we identify one such simulated clock and identify its effectiveness.

\subsection{Simulated Clocks: Hybrid Logical Clocks}
\label{sec:simclock}

Hybrid Logical Clocks(HLC) \cite{HLC} are one such instance of simulated clocks. HLC refines both physical clocks and logical clocks \cite{lamport}. In HLC, each event $e$ is timestamped with $\br{pt.e, l.e, c.e}$, where $pt.e$ is the physical time, $l.e$ is the logical time and $c.e$ is a counter. HLC ensures that the logical clock is always {\em close} to the physical clock.
HLC also preserves the property of logical clocks ($e hb f \Rightarrow hlc.e < hlc.f$), where $hlc.e < hlc.f$ iff $(l.e < l.f \vee ((l.e = l.f) \wedge c.e < c.f)))$.

From the above discussion, if $(l.e = l.f) \wedge (c.e = c.f)$ then this implies that $e$ and $f$ are concurrent. Observe that this is exactly the property required of the simulated clocks.

\subsection{Analytical Model for detecting predicates with simulated clocks (HLC)}

In essence, a quasi-synchronous monitor that uses a simulated clock detects a snapshot if and only if there is a point (in time) common to the local intervals associated with all processes. As a result, any snapshot discovered is always $\epsilon_{app}$-snapshot, for any value of $\epsilon_{app}$. In other words, Precision of detection using a simulated clock is always equal to one. So, we focus only on Recall. 

Recall can be computed as probability of $\epsilon_{app}$-snapshot being detected by the quasi-synchronous monitor, which is equivalent to the event that $\epsilon_{app}$-snapshot has common point of intersection.  This can be done by computing the probability of a snapshot being overlapped normalized by the probability of a snapshot being $\epsilon_{app}$-snapshot. To compute the probability of a snapshot being overlapped, we fix the first interval event happened at process 0 at time 0. Then, for the rest of the event, we compute the latest minimum point of first interval among processes after time 0. The intervals have common intersection if and only if the length of this latest minimum point from 0 is shorter than the length of interval $\ell$.  With this idea, we show in the following Theorem. 

\begin{thm}
Quasi-Synchronous Monitoring has Recall :
$$ Recall =  \frac{f(\ell)}{f(\epsilon_{app}+\ell)}$$
Where $$ f(x) = (1-(1-\beta)^x)^{n-1} $$
\label{thm:hlcrecall}
\end{thm}

Given that we can compute Recall of quasi-synchronous monitoring, we also want to know when majority of true snapshots are found by a quasi-synchronous monitor. That is, given an application configuration, we want to compute the necessary condition of the length of the intervals such that Recall is at least 0.5. 

\begin{thm} 
With a quasi-synchronous monitor, Recall is at least 0.5 if and only if the following inequality holds:
$$\ell \geq \log_{1-\beta}( \frac{2^{1/(n-1)} - 1 }{ 2^{1/(n-1)} - (1-\beta)^{\epsilon_{app}  }} ) $$
\end{thm}
\begin{proof}[Proof Sketch]
Using Recall in Theorem \ref{thm:hlcrecall}, the result follows from   derivation of the following inequality: 

$$  \frac{(1-(1-\beta)^{\ell})^{n-1} }{(1-(1-\beta)^{\epsilon_{app}+\ell})^{n-1}} \geq 0.5 $$
\end{proof}

\subsection{Simulation results for detecting predicates with simulated clocks}

Using a quasi-synchronous monitor in point based predicates is not expected to be effective. Hence, we focus on its use only with interval-based predicates. 

\begin{figure}
\centering
\includegraphics[width=0.5\textwidth]{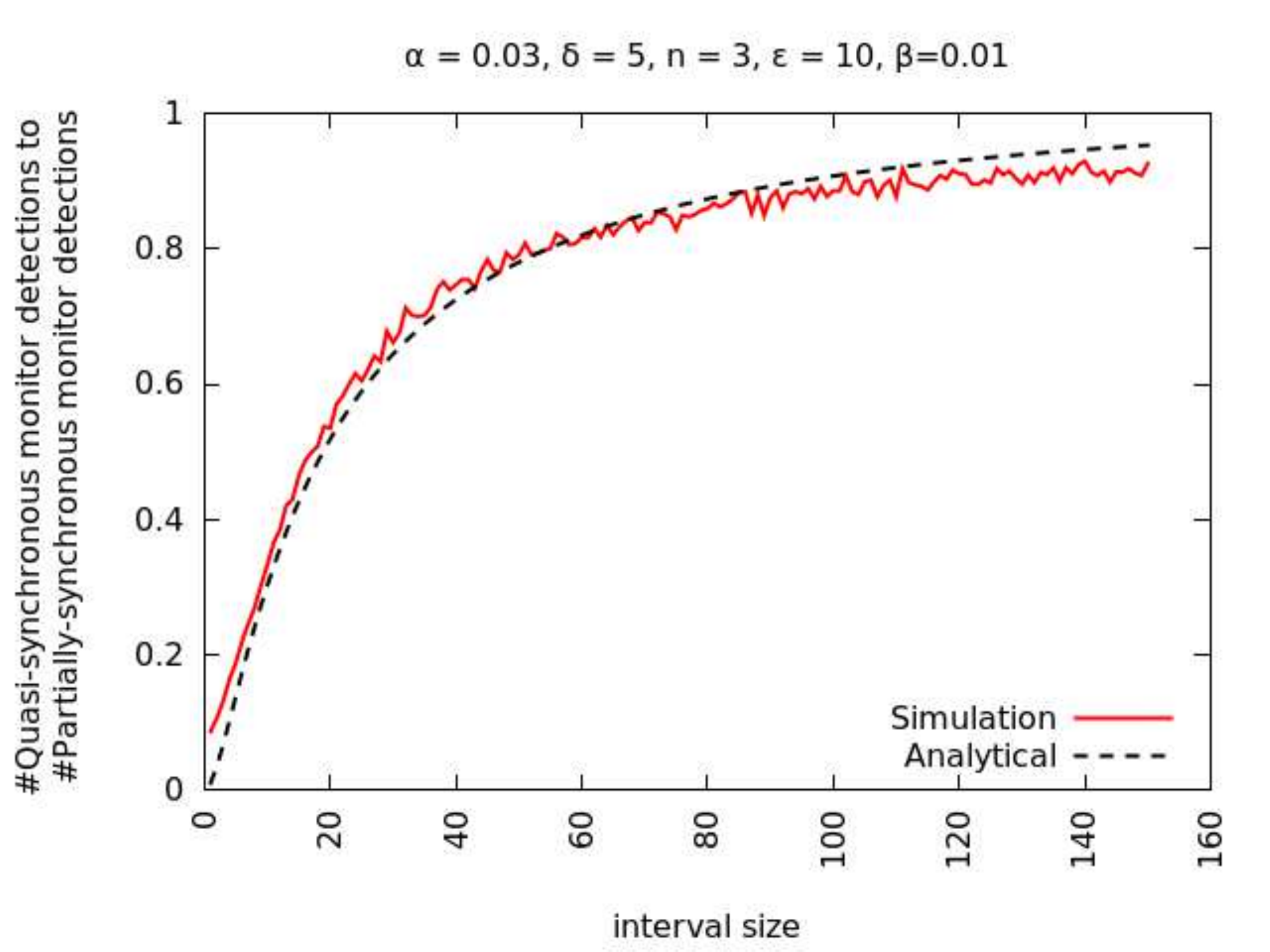}
\caption{Analytical Model vs. Simulation Results: No. of snapshots detected by Quasi-Synchronous Monitor to No. of snapshots detected by Partially Synchronous Monitor}
\label{fig:AnalysisVsSimulation}
\end{figure}

{\bf Effect of Interval length. } \
In Figure \ref{fig:AnalysisVsSimulation}, we present our results for the case where we vary the length of the interval from 1-150. In this figure, if interval length is at least 20 then the probability of finding the consistent snapshot with the quasi-synchronous monitor is 50\% of that computed with partially synchronous monitor. 
%
%
From this observation, we find that for 3 processes when the interval length is approximately $2\epsilon$, the quasi-synchronous monitor is able to detect at least half of the snapshots computed by partially synchronous monitor. Since partially synchronous monitors need vector clocks to detect concurrent events, whereas quasi-synchronous monitors are based on the use of scalar clocks, deploying a quasi-synchronous monitor is expected to be simpler than partially synchronous monitor. We can observe that even with such scalar clocks, quasi-synchronous monitors are able to detect about half of the snapshots identified by partially synchronous monitors. 

Moreover, the analytical model is very close to the simulation model. Hence, one can utilize the analytical model to identify the expected false negative rate to determine whether a quasi-synchronous monitor should be deployed. 

\subsection{Detecting Partial Global Predicates With Quasi-Synchronous Monitors}

In the earlier discussion, we considered the case where the predicate being monitored is a conjunctive predicate involving all processes. Now, we consider the case where the predicate being detected involves only a subset of $p$ processes, $p \leq n$. Instances of such protocols include scenarios where the monitor needs to check if two (given) processes have a token at the same time. 

We analyze the performance of quasi-synchronous monitors in detecting such predicates. We expect quasi-synchronous monitors to perform better with decrease in the size of the subset of processes in the underlying system for which the predicate is evaluated for truthfulness. 

We simulated this by monitoring a system of $n$ processes using a quasi-synchronous monitor and partially synchronous monitor simultaneously, to detect occurrences of partial predicates i.e. if predicate is true for $p$ ($<=n$) processes. When $p=n$ the quasi-synchronous monitor detected about half the number of global predicates detected by the partially synchronous monitor as expected. With decrease in the value of $p$ the number of occurrences detected by the quasi-synchronous monitor started to approach the number of such partial predicates detected by the partially synchronous monitor. From Table \ref{table:HLCpOfn}, we observe that the fraction of predicates detected by HLC increases when the value of $p$ becomes small. 

\vspace*{-\baselineskip}
\begin{table}
\centering
\begin{tabular}{ |c | c| }
  \hline
	p & Fraction of snapshots detected \\
& by quasi-synchronous monitor\\ \hline
	2  & 0.79\\
	3 & 0.68\\
	4  & 0.60\\
	5 & 0.42\\
  \hline  
\end{tabular}
\caption{Partial predicate detection by quasi-synchronous monitor with 5 processes}
\label{table:HLCpOfn}
\end{table}
\vspace*{-\baselineskip}

\section{Omitted Proofs}

\subsection{Proof of Theorem \ref{thm:vcinterval}}
\label{sec:vcinterval-proof}

We first show, $\phi( \epsilon,n,\beta)$,  probability of $hb$-consistent snapshot being $\epsilon$-snapshot for point-based predicate. This is equivalent to computing distribution of $L(c)$ where each interval has length 1. For point-based predicate, the result is as follows and its derivation is provided as a proof. For convenience, we denote $\phi( x ,n,\beta)$ as $g(x)$ representing the length of a point-based predicate snapshot. 

\begin{lem}
For point-based predicate, let $c$ be an $hb$-consistent snapshot. 
The probability of $c$  being  $\epsilon$-consistent (true positive rate) is $\phi( \epsilon,n,\beta) = (1-(1-\beta)^{\epsilon})^{n-1}$
\label{lem:tpr-hbcut}
\end{lem}
\begin{proof}
We first fix one process to have true predicate at time 0.  We define random variable $x_i$ as the first time after time 0 that the predicate is true at process $i$, $ 2 \leq i \leq n$. So, $x_i$ has geometric distribution with parameter $\beta$, i.e., $P(x_i \leq \epsilon) = 1-(1-\beta)^{\epsilon}$. The cut is $\epsilon$-consistent if all points are not beyond $\epsilon$. That is, 
\begin{eqnarray*}
 P(\max_{1 \leq i \leq n} x_i \leq \epsilon) &=& \prod_{i=1}^{n-1} P(x_i \leq \epsilon) \\
& = &(1-(1-\beta)^{\epsilon})^{n-1}
\end{eqnarray*}
\end{proof}

To complete the proof, we calculate probability of $hb$-consistent snapshot being $\epsilon$-snapshot for interval-based predicate of length $\ell$. Using Lemma \ref{lem:tpr-hbcut}, we can obtain the following result.  For convenience, we denote $\phi( x ,n,\beta,\ell)$ as $f(x)$ representing the length of a interval-based predicate snapshot. 

 We simply calculate $P(L(c) \leq \epsilon)$. In this case, $L(c) = \max(\max_i(\{a_i\})-\min_i(\{b_i\})),0)$ by definition of length of snapshot $c$, $L(c)$.  Hence, 
 \begin{eqnarray*}
 P(L(c) \leq \epsilon) &=& P(\max(\max_i(\{a_i\})-\min_i(\{b_i\}),0) \leq \epsilon) \\
   & =&  P(\max_i(\{a_i\}) \leq \epsilon+\ell-1)  \\
	 & =&  g(x+\ell-1)  
\end{eqnarray*}
 The result follows since $ \phi( \epsilon,n,\beta, \ell) = P(L(c) \leq \epsilon)$.

\subsection{Proof of Theorem \ref{thm:uncertainrange}}

The result follows immediately as a Corollary of the following Lemma.
\begin{lem}
For $n > 1$, two inflection points of slopes are at $$ \epsilon_{p1} = \frac{\ln(\frac{3n-4+ \sqrt{5n^2-16n+12}}{2(n-1)^2})}{\ln(1-\beta)}$$ 
and 
$$ \epsilon_{p2} = \frac{\ln(\frac{3n-4-\sqrt{5n^2-16n+12} }{2(n-1)^2})}{\ln(1-\beta)}$$
\label{lem:two}
\end{lem}
\begin{proof} 
Solve a system of equations of the third order derivative of $\phi( \epsilon,n,\beta)$ with respect to $\epsilon$ by definition of inflection points of slopes.
\end{proof}

To complete the proof, we take ratio from Lemma \ref{lem:two}, and compute the limit as $n \rightarrow \infty$.

\subsection{Proof of Theorem \ref{thm:hvcpr}}
Precision can be calculated as follows. If $\epsilon_{mon} < \epsilon_{app}$, then Precision is 1 since all $\epsilon_{mon}$-snapshots are $\epsilon_{app}$-snapshots, but not vice versa. If $\epsilon_{mon} > \epsilon_{app}$, then Precision can be calculated as probability of $\epsilon_{mon}$-snapshot being $\epsilon_{app}$-snapshot. In other words, Precision is probability of a snapshot has length of $\epsilon_{app}$ given that the snapshot is of length $\epsilon_{mon}$.  Therefore, let $L(c)$ be length of snapshot $c$; Precision is given by

\begin{eqnarray*}
 Precision&=& P(L(c) \leq \epsilon_{app} | L(c) \leq \epsilon_{mon}) \\
 &=& \frac{P(L(c) \leq \epsilon_{app} \textrm{ and } L(c)) \leq \epsilon_{mon} }{P(L(c) \leq \epsilon_{mon})} \\
 &=& \frac{P(L(c) \leq \min(\epsilon_{app},\epsilon_{mon}))}{P(L(c) \leq \epsilon_{mon})} \\
 &=& \frac{f(\min(\epsilon_{app},\epsilon_{mon}))}{f(\epsilon_{mon})} 
\end{eqnarray*}

Similarly, Recall is probability of a snapshot being of length $\epsilon_{mon}$ given that the snapshot is of length $\epsilon_{app}$.  By straightforward algebra,  we have the following result.

$$ Recall = \frac{f(\min(\epsilon_{app},\epsilon_{mon}))}{f(\epsilon_{app})}  $$

\subsection{Proof of Theorem \ref{thm:hvcaccurate}}

We fix $\epsilon_{app}$ and then we bound the target $\epsilon_{mon}$. If $\epsilon_{app} < \epsilon_{mon}$, then by Theorem \ref{thm:hvcpr} Precision is $$ (\frac{1-(1-\beta)^{\epsilon_{app}+\ell-1}}{1-(1-\beta)^{\epsilon_{mon}+\ell-1}})^{n-1}$$ We want precision to be at least $\eta$ where $0 \leq \eta \leq 1$. We establish an inequality: $$  (\frac{1-(1-\beta)^{\epsilon_{app}+\ell-1}}{1-(1-\beta)^{\epsilon_{mon}+\ell-1}})^{n-1} \geq \eta $$ 
The results follow from solving the inequality for both Precision and Recall cases. 

\subsection{Proof of Theorem \ref{thm:hvcphasetransition}}

We use the same technique as point of inflections of slopes to obtain the phase transition. The phase transition is defined as the point that maximizes concavity or convexity. We obtain by solving an equation given by third order derivative of the bound in Theorem \ref{thm:hvcaccurate} equal to zero. We use Computer Algebra, WolframAlpha, to derive this expression.

\end{document}